\def\year{2019}\relax
\documentclass[letterpaper]{article} %DO NOT CHANGE THIS
\usepackage{aaai19}  %Required
\usepackage{url}  %Required
\usepackage{graphicx}  %Required
\nocopyright

\usepackage{amsmath,amssymb,amsthm}
\usepackage{subfigure}
\newtheorem{theorem}{Theorem}
\title{Two Can Play That Game: An Adversarial Evaluation of a Cyber-alert Inspection System}
\author{Ankit Shah$^\dagger$\thanks{Ankit and Arunesh contributed equally to this work.}, Arunesh Sinha\textsuperscript{\textparagraph}\footnotemark[1]{}, Rajesh Ganesan$^\dagger$, Sushil Jajodia$^\dagger$, Hasan Cam$^\ddagger$\\
$^\dagger$George Mason Univerity, USA, \textsuperscript{\textparagraph}University of Michigan, USA, $^\ddagger$Army Research Lab, USA\\
ashah20@gmu.edu, arunesh@umich.edu, \{rganesan,jajodia\}@gmu.edu, hasan.cam.civ@mail.mil}
\date{}

\begin{document}

\maketitle

\begin{abstract}
    Cyber-security is an important societal concern. Cyber-attacks have increased in numbers as well as in the extent of damage caused in every attack. Large organizations operate a Cyber Security Operation Center (CSOC), which form the first line of cyber-defense. The inspection of cyber-alerts is a critical part of CSOC operations. A recent work, in collaboration with Army Research Lab, USA proposed a reinforcement learning (RL) based approach to prevent the cyber-alert queue length from growing large and overwhelming the defender. Given the potential deployment of this approach to CSOCs run by US defense agencies, we perform a red team (adversarial) evaluation  of this approach. Further, with the recent attacks on learning systems, it is even more important to test the limits of this RL approach. Towards that end, we learn an adversarial alert generation policy that is a \emph{best response} to the defender inspection policy. Surprisingly, we find the defender policy to be quite robust to the best response of the attacker. In order to explain this observation, we extend the earlier RL model to a game model and show that there exists defender policies that can be robust against any adversarial policy. We also derive a competitive baseline from the game theory model and compare it to the RL approach. However, we go further to exploit assumptions made in the MDP in the RL model and discover an attacker policy that overwhelms the defender. We use a \emph{double oracle} approach to retrain the defender with episodes from this discovered attacker policy. This made the defender robust to the discovered attacker policy and no further harmful attacker policies were discovered. Overall, the adversarial RL and double oracle approach in RL are general techniques that are applicable to other RL usage in adversarial environments.
    %Hence, our red team evaluation improved robustness of this cyber-alert management approach.
\end{abstract}

\section{Introduction}
In this era of truly pervasive computing, cyber-security has emerged as a major concern. Cyber-attacks have crippled large hospitals~\cite{hospital} and resulted in stolen sensitive information from large companies as well as defense agencies~\cite{defense}. Most large organization, including defense agencies, operate a Cyber Security Operation Center (CSOC). A CSOC is a team of specialized analysts, engineers and responders, responsible for maintaining and improving cyber-security. A critical task in the CSOC is to inspect cyber-alerts generated by various sensing software such as SNORT or tools such as ArcSight. Given the high false positive rates of cyber-alerts~\cite{falsepositive}, it is important to screen the alerts effectively to identify any real attack signal from these alerts, which requires maintaining the queue length of alerts within acceptable limits.

A recent work~\cite{Shah2018}, in collaboration with Army Research Lab, USA (ARL), proposed a reinforcement learning based approach (called CSOC-RL) to manage inspection of cyber-alerts at ARL. The model in the CSOC-RL work is quite detailed to match actual operations at ARL, but, the adversary was modeled as a fixed stochastic agent. Distinct from CSOC-RL, another work~\cite{schlenker2017don} uses a game theoretic approach for cyber-alerts inspection, but, the model used simplifying assumptions such as fixed number of analysts and alerts arriving every hour, as well as a single shot interaction. 
In this work, we collaborate with ARL to perform a red team evaluation of the CSOC-RL approach, using both empirical and theoretical techniques.

Our \emph{first contribution} is to pose the strategy learning problem of an attacker as a reinforcement learning (RL) problem itself, given a fixed deployed defender policy. The CSOC-RL model is based on an underlying queuing process of cyber-alerts  (explained in \S~\ref{sec:background}). The defender has a base number of analysts and a given budget of \emph{additional} analysts. His action is to allocate (or not) additional analysts in every hour depending on the state of system, where the state is roughly the amount of backlog of alerts. The defender aims to keep the backlog of alerts below a threshold provided by ARL. We model the attacker analogously as choosing to send \emph{extra} alerts over and above the stochastic arrival of alerts based on the underlying queuing process, given a fixed total budget of additional alerts. The attacker aims to push the backlog of alerts above the given threshold. Surprisingly, our experiments reveal that the attacker fails to overwhelm the system (defined quantitatively in \S~\ref{sec:experiments}) when its budget of additional alerts is less than the defender's budget of additional inspections, even though the defender policy was learned against a stochastic adversary. 

In order to understand the failure of the above attacker policy, we formulate a zero-sum game version of the CSOC-RL model as our \emph{second contribution}. This CSOC-GAME model is a partially observable stochastic game but where the total reward is the maximum of the reward in every time step, as opposed to the standard discounted sum of rewards. Building on recent results in stochastic games~\cite{gimbert2016values}, we show that the minimax theorem holds for this non-standard game, which reveals that there exists a defender's policy that is robust against \emph{any} attacker policy. The experimentally observed failure of the attacker's best response to the learned defender policy implies that the defender is learning a policy close to its minimax policy (discussed later in \S~\ref{sec:game}). Digging further into queuing theory, we derive simple rule based defender policies that guarantee certain minimum reward for the defender. We find experimentally that these rule based policies, while simple and interpretable, are inferior in performance to the RL learned defender policy.

Our \emph{third contribution} is a successful attacker policy that works by exploiting a modeling assumption in the CSOC-RL model and a defense against this attack. In the CSOC-RL model, operational considerations restricted the number of inspections allocated by the defender to be in discrete chunks of fixed size. Importantly, the CSOC-RL work also assumed that the adversary sends alerts in discrete chunks of exactly the same size as the defender's inspection chunks, which we exploit in our attack. By relaxing the assumption on the attacker, we find an attacker policy that exhausts the defender budget of additional inspections using a small number of additional alerts, which then allows the attacker to overwhelm the system arbitrarily. 
%Intuitively, the attacker exploits the fixed size of inspection chunks to elicit a disproportionate response to few additional alerts. 
Inspired by the double oracle method from game theory, we retrain the defender using additional episodes from the discovered attacker policy and find that the defender's policy becomes robust to the discovered policy. 
%Intuitively, the relearned defender policy exhibits more patience and allows the backlog to build up more before allocating additional inspections compared to the prior policy. 
In the next double oracle iteration, no harmful attacker policies were discovered, thus, providing evidence that the relearned defender policy is robust to any adversarial generation policy.

\section{Background, Prior Work and Notation}
\label{sec:background}
In this section, we provide a brief summary of the prior work CSOC-RL as well as a brief background about the queuing process used in that work. The arrival of alerts was modeled as a Poisson process with nominal rate of $\lambda_0 = 1919$ alerts/hour, which was chosen based on past data from the CSOC that was studied in that work. The adversary was modeled as a fixed stochastic adversary that changed the actual alert arrival rate $\lambda$ ($\lambda\geq \lambda_0$) according to an unknown stochastic distribution.  This CSOC-RL work modeled the first level of inspection in a CSOC. The first level of inspection is a fast inspection that decides the severity of the alert and whether a followup second detailed inspection is required. Given the almost same steps for all alerts in this first level of inspection, the inspection time for every alert is same. Based on the arrival nominal rate $\lambda_0$ and the time to service an alert, a nominal number of analysts were chosen so that the aggregate \emph{nominal} service rate of alert was $\mu_0 = 1920$ alerts/hour. However, the actual service rate $\mu$ varied stochastically with an unknown distribution (always $\leq \mu_0$) due to factors such as analyst absenteeism, failure of sensors, etc. We skip the details of analyst scheduling in~\cite{Shah2018}, as that is not required for this exposition; for completeness, this detail is in the appendix.

\textbf{Background on queuing theory}: The above model with fixed rates ($\mu = \mu_0, \lambda = \lambda_0$) is exactly a M/D/1/FCFS queue (this notation is the standard Kendall notation from queuing theory). M stands for Poisson arrival, D for deterministic service time, 1 for the number of servers (here all analyst are clubbed together as one server), and FCFS means that the alerts are inspected on a first come first serve basis. The M/D/1/FCFS queue has been studied a lot and can be viewed as a discrete time Markov chain with infinite state space $\{0,1,\ldots\}$ that represents the queue length. The transition probability of this Markov chain depends on $\lambda_0, \mu_0$.  Let $A_t, Z_t$ be the random variable that denotes the number of arrivals and number of alerts serviced in the $t^{th}$ hour. With fixed nominal rate, $P(Z_t=\mu_0)=1$ (deterministic service) and $P(A_t = n) = \frac{\lambda_0^n \exp(-\lambda_0)}{n!}$ (Poisson distribution). Given queue length $b_{t-1}$ at time $t-1$, the transition probability can be expressed as a function $h$ of $\lambda_0,\mu_0$ as follows $P(b_t~|~b_{t-1}) = P(A_t - Z_t = b_t - b_{t-1}) = h(\lambda_0, \mu_0, b_t - b_{t-1})$. An important aspect of M/D/1/FCFS queue is the ratio $\rho_0 = \frac{\lambda_0}{\mu_0}$. The queue is stable with finite expected queue length only when $\rho_0 < 1$. With the given nominal rates we have $\rho_0 = 0.999479$.

\textbf{MDP of the CSOC-RL problem}: However, note that with varying $\lambda, \mu$ (with unknown distribution) the CSOC-RL model is not exactly a M/D/1/FCFS queue. Yet, it can still be described by a Markov decision process with the unknown transition probability as $P(b_t~|~b_{t-1}) = P(A_t - S_t = b_t - b_{t-1}) = \int h(\lambda,\mu, b_t - b_{t-1}) p(\lambda,\mu) \; d\lambda \; d\mu$, where $p(\lambda,\mu)$ is the joint (unknown) density of the randomly varying $\lambda, \mu$. As the transition probability is unknown, 
the CSOC-RL work modeled the defender's problem as a RL problem. Note that as $\lambda \geq \lambda_0$ and $\mu \leq \mu_0$ the instantaneous $\rho = \frac{\lambda}{\mu}$ can be $> 1$. Thus, the defender is provided with a total of $X=28,800$ additional inspections to be used over $N$ time steps, and the $\lambda, \mu$ are so controlled so that the additional alerts (counted as total  additional alerts due to higher $\lambda$ and fewer inspections due to lower $\mu$) is also not more than $X$. The defender-adversary interaction was modeled over $N = 336$ hours (two weeks) as the staffing changes every two weeks. In every hour, the defender could call up to 10 extra analysts, which translated to at most $E = 2400$ additional inspections per hour. Also, an analyst has to be allocated for a minimum of 15 minutes, which translated to a discrete allocation of additional inspection in chunks of $M = 60$ alerts.
Formally, the RL problem was modeled with 
\begin{itemize}
  \item\emph{State}: $s \in S$ is a tuple $s = \langle b, n, x\rangle$, where $b$ is the backlog of alerts at the end of time interval $t$, $n$ is the number of time intervals remaining, $x$ is the resources remaining for the defender. The initial budget for the defender is $X$. The initial value of $n$ is the time horizon $N$. 
  \item \emph{Action}: The defender action is to allocate $d$ ($d \leq x$) additional inspection resources at the start of the time interval. $d$ is a multiple of $M$ and an integer between $0$ and $E$.
  \item \emph{Transition}: In the next state, $n$ decreases by one, $x$ decreases by $d$, and the next $b'$ is given by $P(b'~|~s,d) = P(A_t - Z_t = d + b'-b) = \int h(\lambda,\mu, d+ b' - b) p(\lambda,\mu) \; d\lambda \; d\mu$, where $A_t$, $Z_t$ are defined above.
  \item \emph{Rewards}: The immediate reward has two terms. The first is due to the cost incurred by the defender from the backlog after allocating additional inspections given by $C(s,d) = f(b - d)$, where function $f$ normalizes the cost to lie between $[0,1]$ with the value increasing (not strictly) with its argument. The second term is $q(x,n)$ which incentivizes preserving additional resources per time remaining, i.e., $q$ increases with increasing ratio $x/n$ and normalized to lie in $[0,1]$. Thus, the immediate reward is $- f(b - d) + q(x,n)$. 
\end{itemize}
Details of the RL training, the function $q$, and the size of the problem are provided in the CSOC-RL paper.

\textbf{Measuring performance}: While the RL model described above includes reward for preserving budget, the effective reward for the defender is only from backlog. The purpose of the term $q$ for preserving budget was to converge quickly to learn to not exhaust all budget. Thus, as presented in the CSOC-RL work, we also show rewards only for the backlog term.  In more details, the function $f(x)$ is a piecewise linear function defined as
$$
f(x) = 
\begin{cases}
    0 & \mbox{for } x \leq 1175 \\
    \frac{x - 1175}{4350 - 1175} & \mbox{for } 1175 < x < 4350 \\
    1 & \mbox{for } 4350 \leq x
\end{cases}
$$
The 1175 (4350) alerts corresponds to one hour (four hours) worth of average wait time to inspect an alert (AvgTTA) for the CSOC studied in the CSOC-RL paper. The results in the CSOC-RL paper show the backlog in terms of AvgTTA, which we elaborate on further in the experiments section. From our conversation with ARL, the defender's aim is to keep the maximum AvgTTA over the time horizon $N$ as low as possible, with anything over 4 hours being unacceptable.
%The RL approach maximizes the value of each state (state includes time in this model), thus, it aims (roughly) to keep the AvgTTA for each hour as low as possible (and not just the maximum AvgTTA over the time horizon).

\section{Adversarial Evaluation Methodology}
In this section, we present our approach to the red team evaluation. The approach has three distinct parts that are presented in the sub-sections.
\subsection{Adversarial RL}
\label{sec:advRL}
Recall that the attacker in the CSOC-RL work was a fixed stochastic agent that changed the actual alert arrival rate $\lambda$ ($\lambda\geq \lambda_0$) according to an unknown stochastic distribution. We make the attacker truly adversarial by allowing him to send his choice of a number of alerts in every time step subject to a total budget constraint $Y$. The attacker's optimal attack problem can be set-up as a RL problem itself. The MDP of the adversarial RL problem is described as:
\begin{itemize}
  \item\emph{State space}: $s \in S$ is a tuple $s = \langle b, n, x, y\rangle$, where $b,n,x$ are same as for the defender MDP. $y$ is the additional alerts remaining for the attacker. The initial budget for the attacker is $Y$.
  \item \emph{Action space}: The adversary action is to send $a$ ($a \leq y$) additional alerts at the start of the time interval.  $a$ is a multiple of $M$ and an integer between $0$ and $Y$.
  \item \emph{Transition}: In the next state, $n$ decreases by one, $x$ decreases by $d$ ($d$ given by the fixed and known defender policy), $y$ decreases by $a$, and the next $b'$ is given by $P(b'~|~s,a) = P(A_t - Z_t - d +a = b'-b) = \int h(\lambda_0,\mu, b' - b + d - a) p(\mu) \; d\mu$.
  \item \emph{Rewards}: The immediate reward has two terms. The first is exactly the cost incurred by the defender $C(s,d) = f(b - d)$. The second term is $q(y,n)$ which incentivizes preserving additional alerts per time remaining, i.e., $q$ increases with increasing ratio $y/n$ and normalized to lie in $[0,1]$. Thus, the immediate reward is $f(b - d) +  q(y,n)$. 
\end{itemize}

Few points to note about the adversary model are: (1) the adversary is very powerful as it has complete knowledge of the backlog $b$ and the defender state $x$, (2) the arrival rate $\lambda$ is fixed to $\lambda_0$, since the additional alerts are all controlled by the attacker but the $\mu$ still varies randomly, (3) the adversary has no hard bound on the number of additional alerts per hour (like $E$ for defender), but, the $q$ function acts as a soft bound for the number of additional alerts/hour, and (4) the adversary model assumed here allows attacks \emph{only} by sending additional alerts.

As we show later in experiments, the defender policy learned from the prior RL approach is robust to this attack when $Y \leq X$. This is surprising as learning methods, including RL, have been shown to be vulnerable to attacks. In order to understand and explain this robustness, we analyze the defender-adversary interaction in a game theory model in the next sub-section.
\subsection{Game Theoretic Model}
\label{sec:game}
%Game Model. Show existence of value. Show that value is bad if budget of adversary is more. Derive the rule based policy when budget of adversary is lower with fixed $\lambda,\mu$.

We start by presenting a unified model of a two player zero-sum repeated games formulated in a recent book by~\cite{mertens_sorin_zamir_2015}. Our CSOC-GAME will be presented as an instance of such games. A zero-sum unified repeated game with signals $(S, I, J,C,D,\pi, q, g)$
is defined by a set of states $S$, two finite sets of actions $I$ (for player 1) and $J$ (for player 2), two
sets of signals $C$ (for player 1) and $D$ (for player 2), an initial distribution $\pi \in \Delta (S \times C \times D)$, and a
transition function $q$ from $S \times I \times J$ to $ \Delta(S \times C \times D)$, where $\Delta(\mathcal{S})$ denotes the set of probability distributions on given set $\mathcal{S}$. The reward function for player 1 is given by function
$g: X \times I \times J \rightarrow [0, 1]$. The reward for player 2 is $-g$. At each stage $t$ players choose actions $i_t$ and
$j_t$ and a triple $(s_{t+1}, c_{t+1}, d_{t+1})$ is drawn according to $q(s_t, i_t, j_t)$ (for $t=1$ the draw is according to $\pi$), where
$s_t$ is the current state, inducing the signals $c_{t+1}, d_{t+1}$ received by the players and the
state $s_{t+1}$ at the next stage. Player 1's history at stage $t$ is $c_1, i_1, \ldots, c_{t-1},i_{t-1}, c_t$ and similarly player 2's history is $d_1, j_1, \ldots, d_{t-1},j_{t-1}, d_t$. The history of the game is $c_1,d_1,i_1,j_1, \ldots, c_{t-1},d_{t-1},i_{t-1},j_{t-1}, c_t,d_t$. Given perfect recall, a behavioral strategy for player 1 is a sequence
$\sigma = (\sigma_t)_{t\geq1}$, where $\sigma_n$, the strategy at stage $t$, is a mapping from possible histories
to $\Delta(I)$, with the interpretation that $\sigma_t(h^1)$ is the mixed action used by player 1 after its history $h^1$. Similarly, a behavioral strategy for player 2 is a sequence
$\tau = (\tau_t)_{t\geq1}$. This game model is very general and a suitable choice of signal space can model repeated and stochastic games with perfect or imperfect information. Given the above model, an evaluation function maps infinite game histories to total reward. In typical repeated game, this evaluation function is a discounted sum of the per stage rewards given by $g$.

The defender adversary interaction is a game on top of an underlying stochastic process of arrival and processing of cyber-alerts, with varying rates of arrival $\lambda$ and service $\mu$ (stated in \S~\ref{sec:background}). 
%For exposition, we call the alerts from the underlying stochastic process as base alerts and the alerts sent by the adversary as additional alerts; the defender cannot distinguish between these types of alerts. 
%As stated in Section~\ref{sec:background}, the base alerts arrival is effectively described by a Poisson process with rate $\lambda$ (possibly varying), and the service rate is $\mu$ (also possibly varying). 
%However, there are random variations in $\lambda$ and $\mu$, thus, the positive random variables $\delta \lambda_t$ and  $\delta \mu_t$ describe the effective rate $\lambda_t = \lambda + \delta \lambda_t$ and $\mu_t = \mu - \delta \mu_t$ at time step $t$.
%It is assumed that $\mu > \lambda$, or else the queue will be unstable~\cite{}.
In order to model the defender attacker interaction as a unified repeated game model, we remove two heuristic choices made in the RL models. We first remove the term $q$ in the defender reward as that term was used only for faster convergence. Next, we remove term $q$ in the adversary reward, and instead place a hard bound $E$ on the number of additional alerts per hour.
The game model is as follows:
\begin{itemize}
\item Player 1 is attacker and player 2 is defender. 
\item \emph{State space}: Each state $s \in S$ is a tuple $s = \langle b, n, x, y \rangle$ with the exact same specification as the MDP of the adversary RL.
\item \emph{Action spaces} $I$ and $J$: The adversary action is to send $a$ additional alerts over at the end of the time interval. The defender action is to allocate $d$ additional inspection resources at the start of the time interval. $a$ and $d$ are both multiples of $M$, and both are integers between $0$ and $E$.\footnote{The game model is well-defined when the action space does not change over time. The budget constraint on action is indirectly imposed in the state transition using the $\min$ functions.}
\item \emph{State transition}: Recall that $A_t$ is the random number of alerts arriving due to the underlying Poisson process with random rate $\lambda$ in the time interval $t$. Similarly, $Z_t$ is the number of alerts processed at random service rate $\mu$. Given $s = \langle b, n, x, y \rangle$ with $n > 0$ and actions $d$ and $a$ the resultant state $s' = \langle b', n', x', y' \rangle$ satisfies $n' = n - 1$ if $n > 0$, $x' = \max(0,x - d)$, $y' = \max(0,y - a)$, and $b' = b - \min(d,x) + \min(a,y) + A_t - Z_t$. Let $E$ denote $b' - b + \min(d,x) - \min(a,y)$. It can be seen that $P(s'~|~s,d,a) = P(b'~|~s,d,a) = P(A_t - Z_t = E) = \int h(\lambda_0,\mu, E) p(\mu) \; d\mu$. States with $n = 0$ are sink states.
\item \emph{Time horizon}: While the game is of finite horizon $N$, we model it as infinite horizon by fixing the rewards for both players in any state with $n = 0$ to $0$ (see the next item).
\item \emph{Rewards}: The cost incurred by the defender is the backlog after allocating additional inspections, i.e., $C(s,d) = f(b - d)$. The immediate reward for players is:
\begin{itemize}
\item Attacker: $g(s, d, a) = C(s,d)$ when $n > 0$, else $0$
\item Defender: $-g(s, d, a) = -C(s,d)$ when $n > 0$, else $0$
\end{itemize}
 The game is clearly zero-sum.
\item \emph{Signal space} $C$ and $D$: Both players observe separate signals after their actions. The attacker has complete observation, thus, the attacker's signal after the current time step is the next state $s'$ and $d,a$. The defender does not observe the attacker's action $a$ and the $y$ part of the state. However, the defender receives signals for these, which is the backlog $b'$. Thus, the signal for the defender is $b', n', x', d$. The signal probability for the defender is defined exactly like the state transition probability $P(b'~|~s,d,a)$. 
\item \emph{History of signals}: At any time point $T$ the history observed by the adversary is a history $s_t,d_t,a_t$ for all $t \leq T$. For the defender, the observed history is $b_t,n_t,x_t,d_t$ for all $t \leq T$. A play of the game till time $T$ is defined as $s_t,d_t,a_t$ for all $t \leq T$ (this is same as observed history of attacker as adversary observes all the past). All infinite plays of the game are denoted by the set $H_{\infty}$.
\item \emph{Strategy}: As defined earlier, strategies are functions of observed player histories to mixed actions. For attacker, strategy $\sigma$ is a function from all past states and both players action. For defender, strategy $\tau$ is a function of all past $b, n, x$ and $d$. 
\item \emph{Evaluation function}: In contrast to standard stochastic game, the defender (ARL here) only cares about the maximum length of queue of alerts over the time intervals. Thus, the long term reward for the defender is the inf of the rewards over all time intervals (note that defender stage reward is negation of cost, hence, inf captures worst backlog over time). For the attacker, due to the antagonistic nature of interaction, the long term reward is then $sup(h) = \sup_{t\geq 1} g(s_t,d_t,a_t)$ for $h \in H_{\infty}$
\end{itemize}

Fortunately, even for such a complex game a minimax theorem holds for sup evaluation~\cite{gimbert2016values} (which is not true for many other evaluation functions), thus, 
$$
\sup_{\sigma} \inf_{\tau} E_{\sigma,\tau}(sup(h)) = \inf_{\tau} \sup_{\sigma} E_{\sigma,\tau}(sup(h))
 = V$$
 The above result can also be interpreted from the defender's perspective that there is a strategy (policy) $\tau^*$ that can achieve value $-V$ irrespective of the strategy used by the adversary. However, the magnitude of the value $V$ is important (e.g., $V = 1$ is not very good for defender). Next, we show that the budgets of the players decide what $V$ will be achieved for the special case of fixed service rate. Fixed service rate is a mild assumption as analyst absentism is rare and accounted for by a buffer of additional analyst time (see the analyst scheduling details in appendix).
 %show that when the adversary has more budget than adversary then $V$ will be high.
 
\begin{theorem} \label{thmonly}
Given fixed $\mu= \mu_0$, the defender can guarantee himself a long term reward of 
(a) at most $-0.95$, if $Y-X > 4800$, and 
(b) at least $(1 + f(B) )(1 - 1/B)^{N\mu/B} - 1$, if $Y \leq X$. Further, the guarantee of case (b) is provided by a simple rule based policy: ($S1$) whenever the backlog exceeds $B$ by $x$, allocate $x$ additional inspections. 
\end{theorem}
\begin{proof}[Proof Sketch]
If $Y - X > 4800$, then the attacker will dump $E$ additional alerts for 14 hours at the start. W.h.p. the number of normal alerts arriving in 14 hours is also high, thus, the queue will build up to 4350 w.h.p. The case $Y \leq X$ involves concepts from queuing theory. Roughly, it is shown that by its own the backlog for the M/D/1 queue will not exceed $B$ with high probability. Next, it is shown that for any trace in which the backlog due to the M/D/1 process does not exceed $B$ and for any attacker policy, the stated simple rule above is enough to ensure that the backlog is within $B$. Thus, w.h.p. the defender cost is bounded to $f(B)$, which provides the required result with low expected reward. For example, $B = 1500$ yields around $-0.3$ utility.
\end{proof}

The analysis above reveals that there exists robust defender polices that are robust to any attacker policy only when the resources (budget) of attacker is lower. Thus, robustness is highly dependent on the resources of the players, that is, the player with a resource advantage wins the game. The observed robustness of the CSOC-RL defender policy in experiments can be explained as the learned defender policy being one of the robust policies (or being close enough). 

\begin{figure*}%
    \centering
    \subfigure[Daily bound worst run]{%
    \label{fig:first}%
    \includegraphics[scale=0.16]{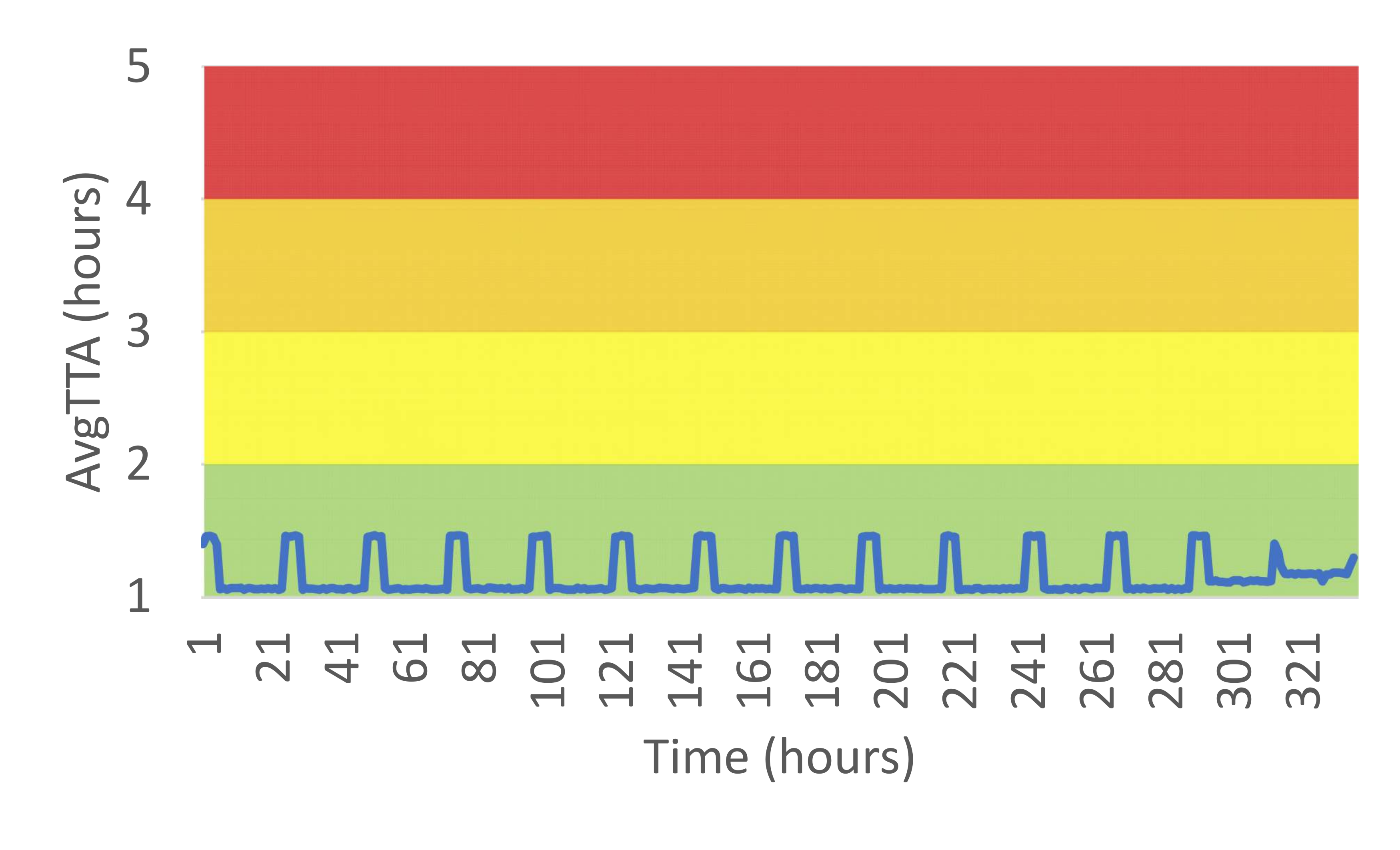}}%
    \quad
    \subfigure[Unbounded worst run]{%
    \label{fig:second}%
    \includegraphics[scale=0.16]{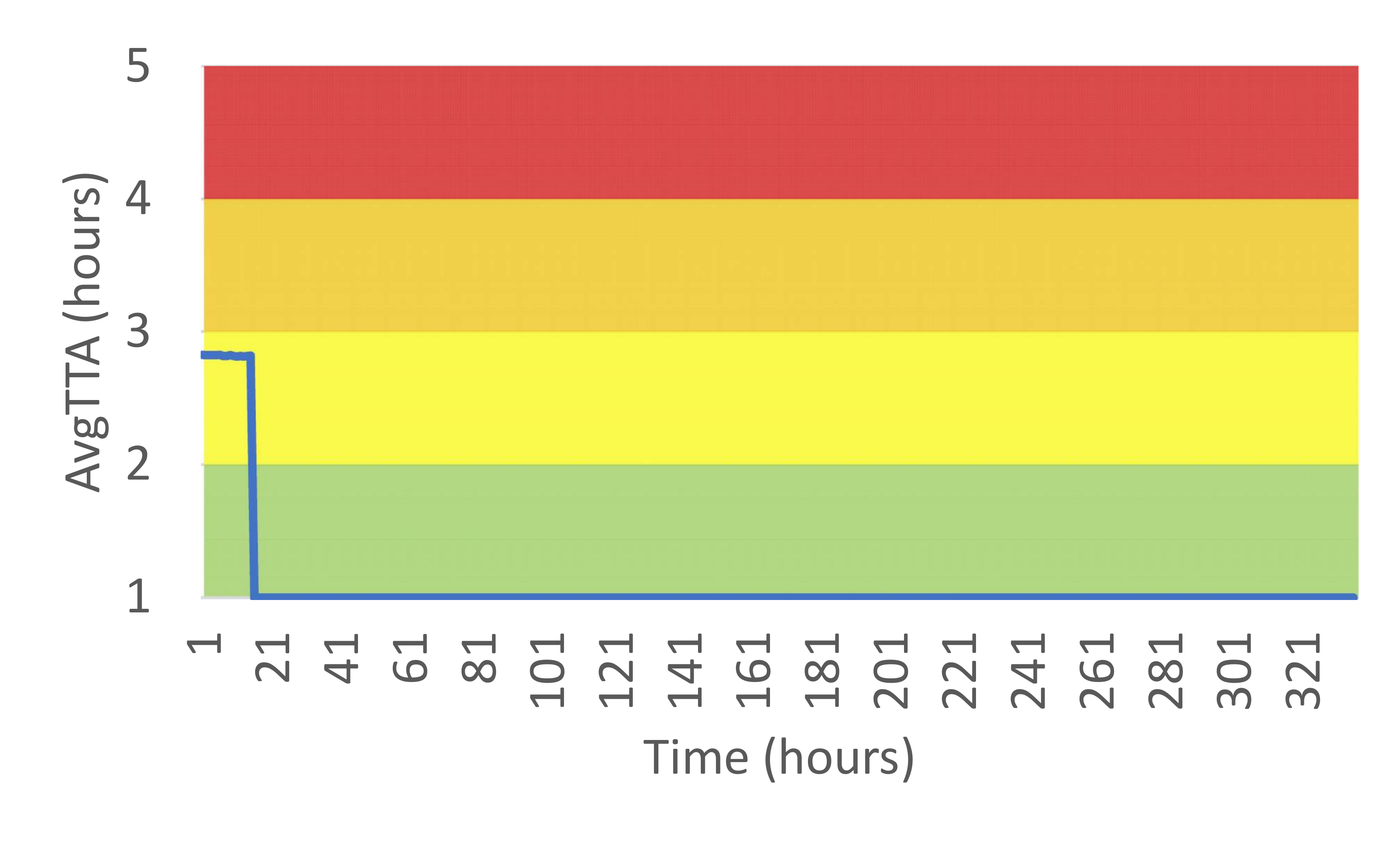}}%
    \quad
    \subfigure[10\% extra budget worst run]{%
    \label{fig:third}%
    \includegraphics[scale=0.16]{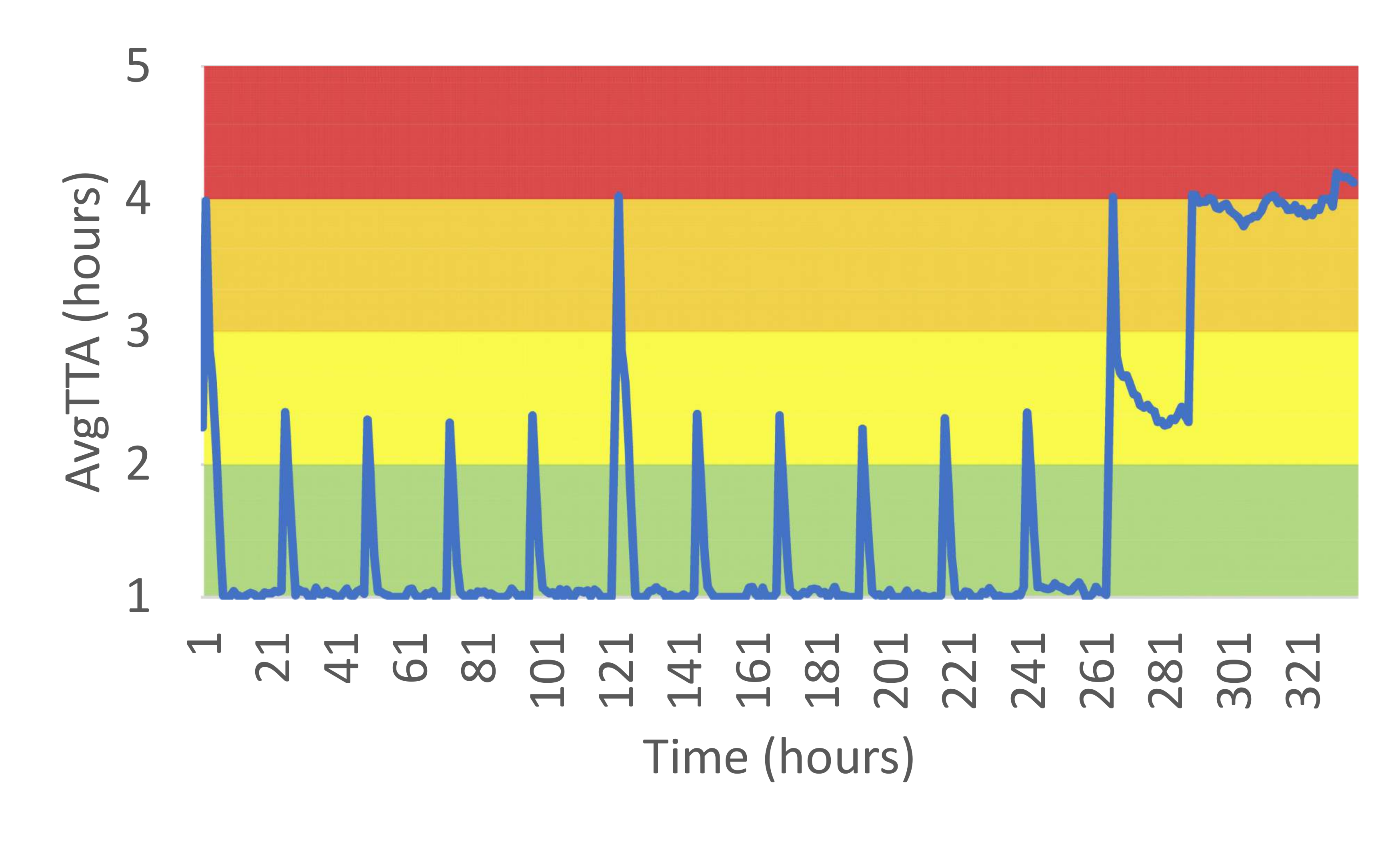}}%
    \qquad
    \subfigure[Daily bound proportions]{%
    \label{fig:fourth}%
    \includegraphics[scale=0.17]{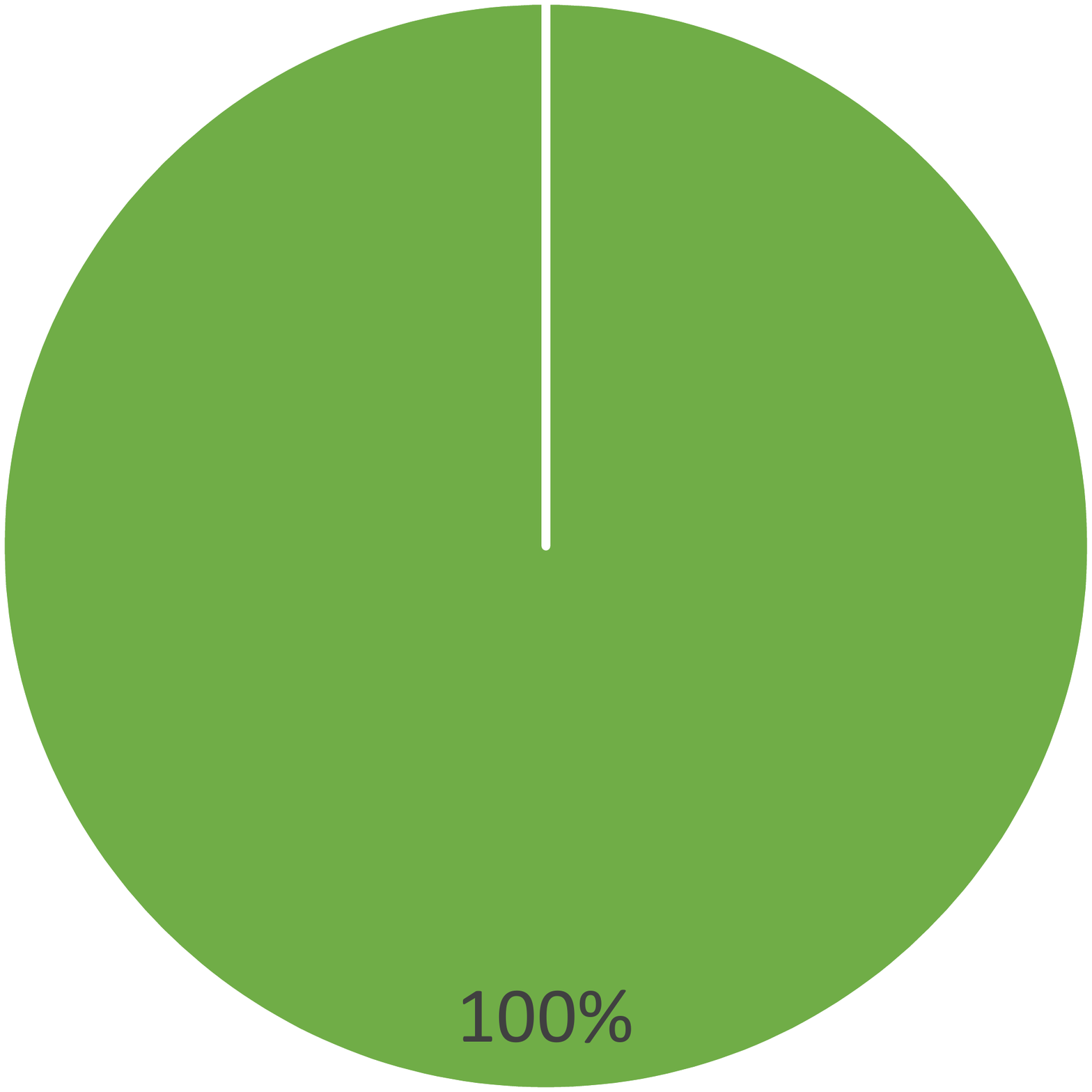}}%
    \qquad\qquad\;
    \subfigure[Unbounded  proportions]{%
    \label{fig:fifth}%
    \includegraphics[scale=0.17]{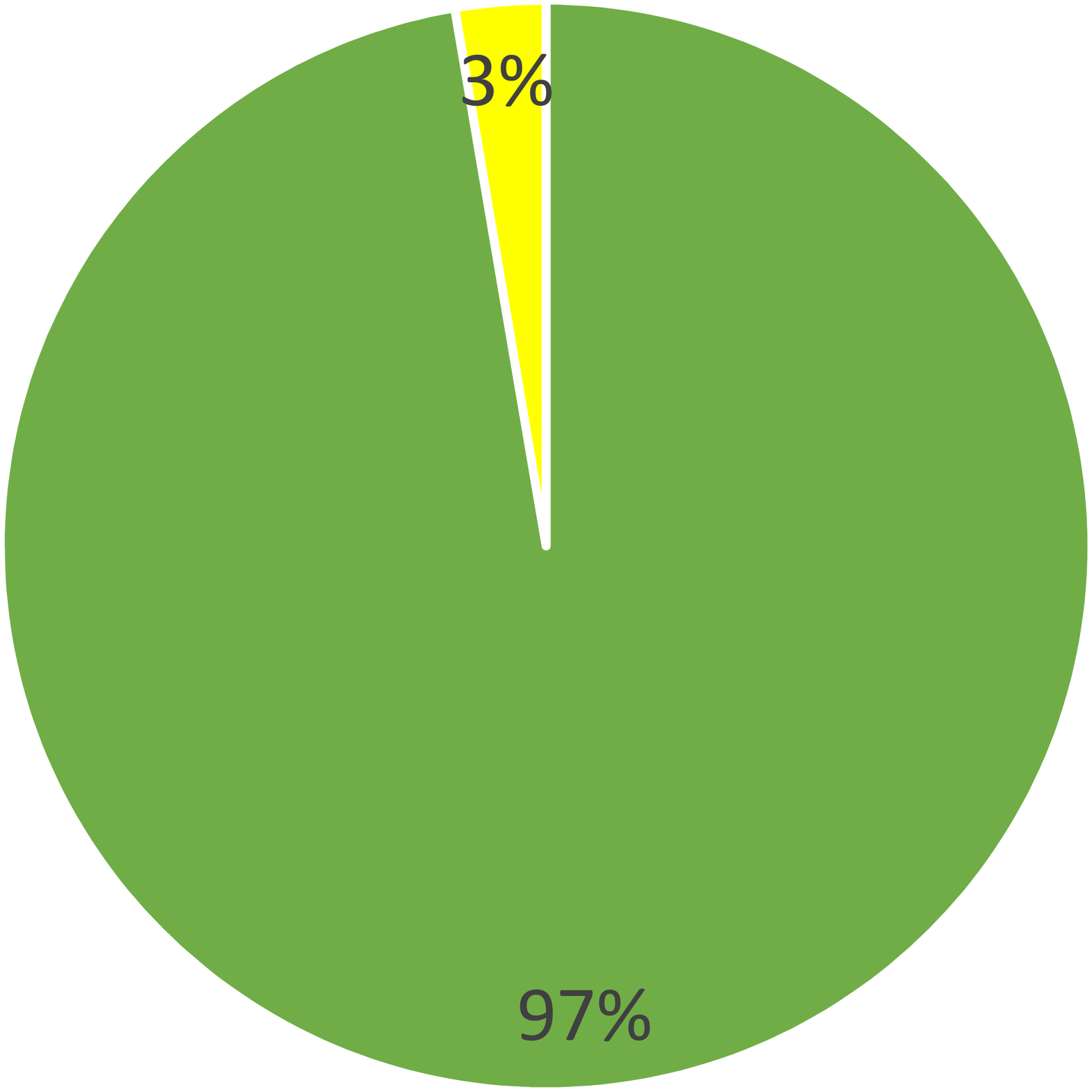}}%
    \qquad\quad\;
    \subfigure[10\% extra budget proportions]{%
    \label{fig:sixth}%
    \includegraphics[scale=0.17]{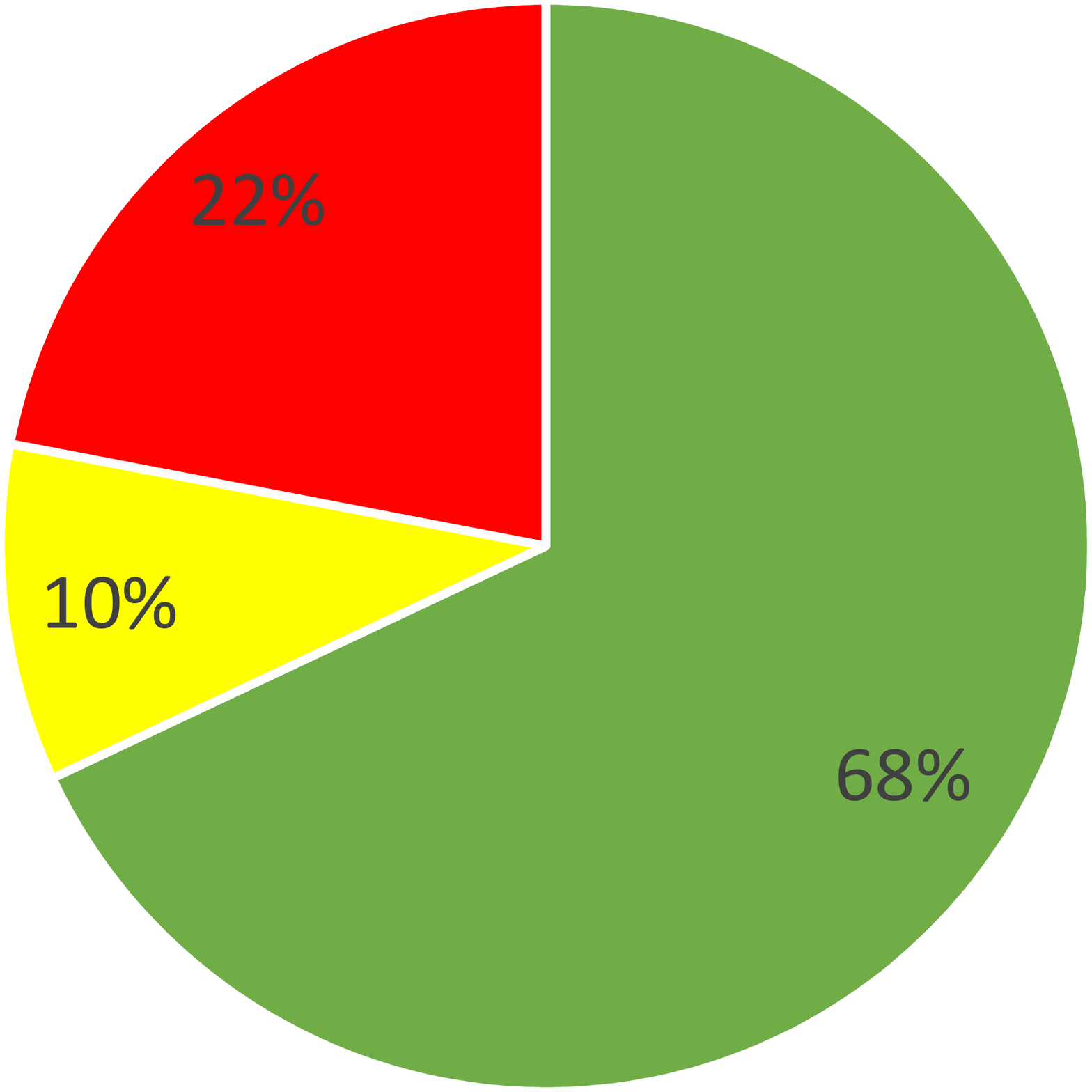}}%
\caption{Results for Adversarial RL attacks} \label{fig:1}
\end{figure*}

Further, the theorem provides a simple rule based baseline policy $S1$. We further propose a more aggressive rule based policy: ($S2$) whenever the backlog exceeds $B$ by $x$, allocate $x + (B-1175)$ additional inspections. $S2$ is more aggressive as it brings the backlog down to 1175 whenever it exceeds $B$. We compare the RL approach against both $S1$ and $S2$ in our experiments.

\subsection{Attacking Modeling Assumptions}
%The last two sub-sections followed the assumption from the CSOC-RL work that the additional inspections and additional alerts arrive in fixed chunks of size 60. 
The last two sub-sections followed the assumption from the CSOC-RL work that the attacker generates alerts with the same discretization in the action-space as the defender, which was fixed to multiples of $M=60$ alerts.
This models a practical constraint for the defender that an additional analyst must be allocated for a minimum amount of time once engaged.
%This models a practical constraint for the defender that an additional analyst has to be allocated for minimum of 15 minutes. 
However, this constraint in the action-space is too restrictive for the attacker. Thus, we relax this constraint for the attacker allowing him to have a finer discretization by sending alerts in multiples of 30. We attack using the same adversarial RL set-up of Section~\ref{sec:advRL} with this new chunk size of 30 for adversary. We show in the experiments that this attack succeeds in building up a large backlog (high AvgTTA). 
The result is explained by the observation that the attacker uses few additional alerts to elicit a disproportionate response of additional resources from the defender, resulting in the defender resources getting exhausted while the attacker still has additional alerts left. Then, the attacker uses the leftover additional alerts to push the backlog to a high value.

The fix for the above attack is inferred from the game theoretic nature of the problem. In particular, this new attack policy is an adversary best response that was not considered when learning the defender RL policy. Thus, using a \emph{double oracle} like approach we retrain the defender by including episodes from this new attacker policy in the defender RL training. We find the new relearned defender policy is robust to the discovered attacker policy, and no new harmful attacker policies were discovered when the adversarial RL approach was used to attack the relearned defender policy. 

\section{Experiments}
\label{sec:experiments}
We explain the experimental setup including our measurement metrics. First, the experimental setup is to simulate the arrival and processing of alerts. The simulation model we use is exactly same as in the previous CSOC-RL work. At a high level, the simulator simulates the underlying queuing process as well as uncertain events such as analysts absentism, etc. Other significant realistic aspects simulated are that additional analysts time is obtained by first making regular analysts work overtime and then bringing in additional analysts, if required. As in the prevous work, we fix the defender budget as 28,800 additional inspections over the $N=336$ hours. For completeness, further details of the experimental set-up is in the appendix.

In the same CSOC-RL work, the authors rely on a metric called average total
time for alert investigation (AvgTTA). The time for alert investigation of an alert is the waiting time in the queue
and the analyst investigation time  after its arrival in the CSOC database.
The AvgTTA/hr is estimated as the average of the time for alert investigation values of all the alerts
that were investigated in an hour. It is a requirement of the CSOC under study that the AvgTTA/hr remain within four hours. One hour or less was determined to be ideal, and anything within 1-4 hours is acceptable. 

\begin{figure*}%
    \centering
    \subfigure[$S1$ policy worst run]{%
    \label{fig:seventh}%
    \includegraphics[scale=0.16]{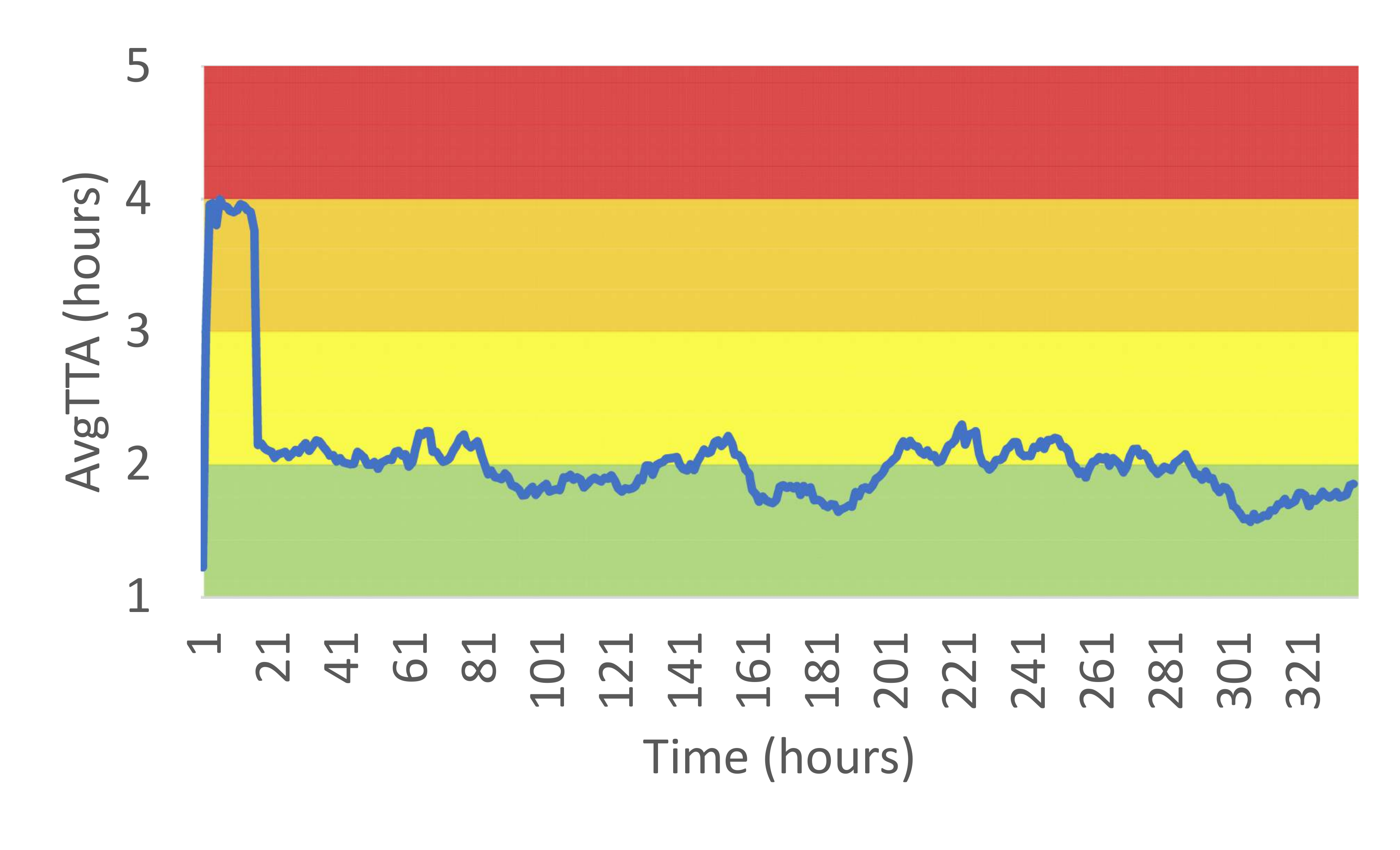}}%
    \quad
    \subfigure[$S2$ policy worst run]{%
    \label{fig:eigth}%
    \includegraphics[scale=0.16]{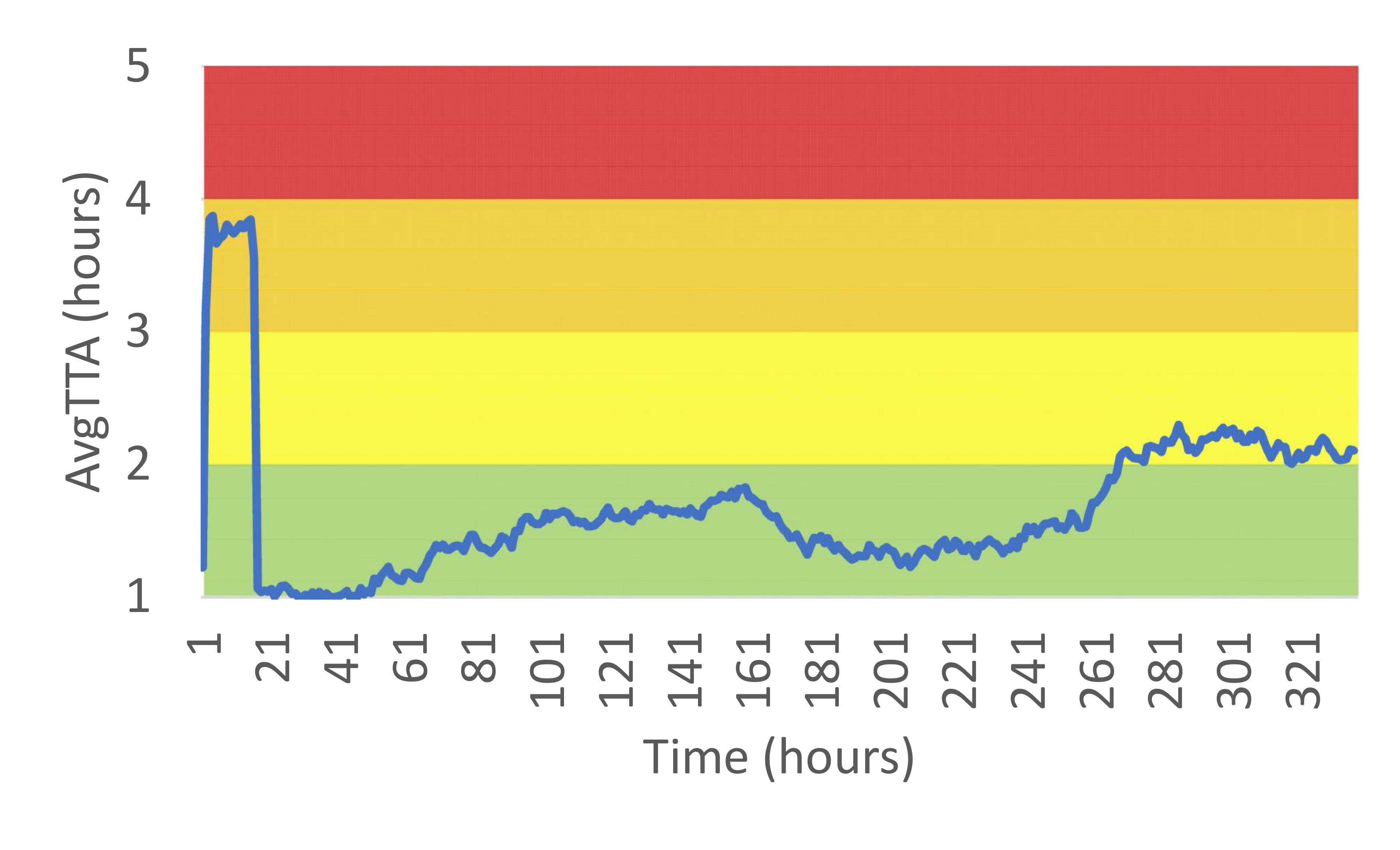}}%
    \quad
    \subfigure[Model assumption attack worst run]{%
    \label{fig:ninth}%
    \includegraphics[scale=0.16]{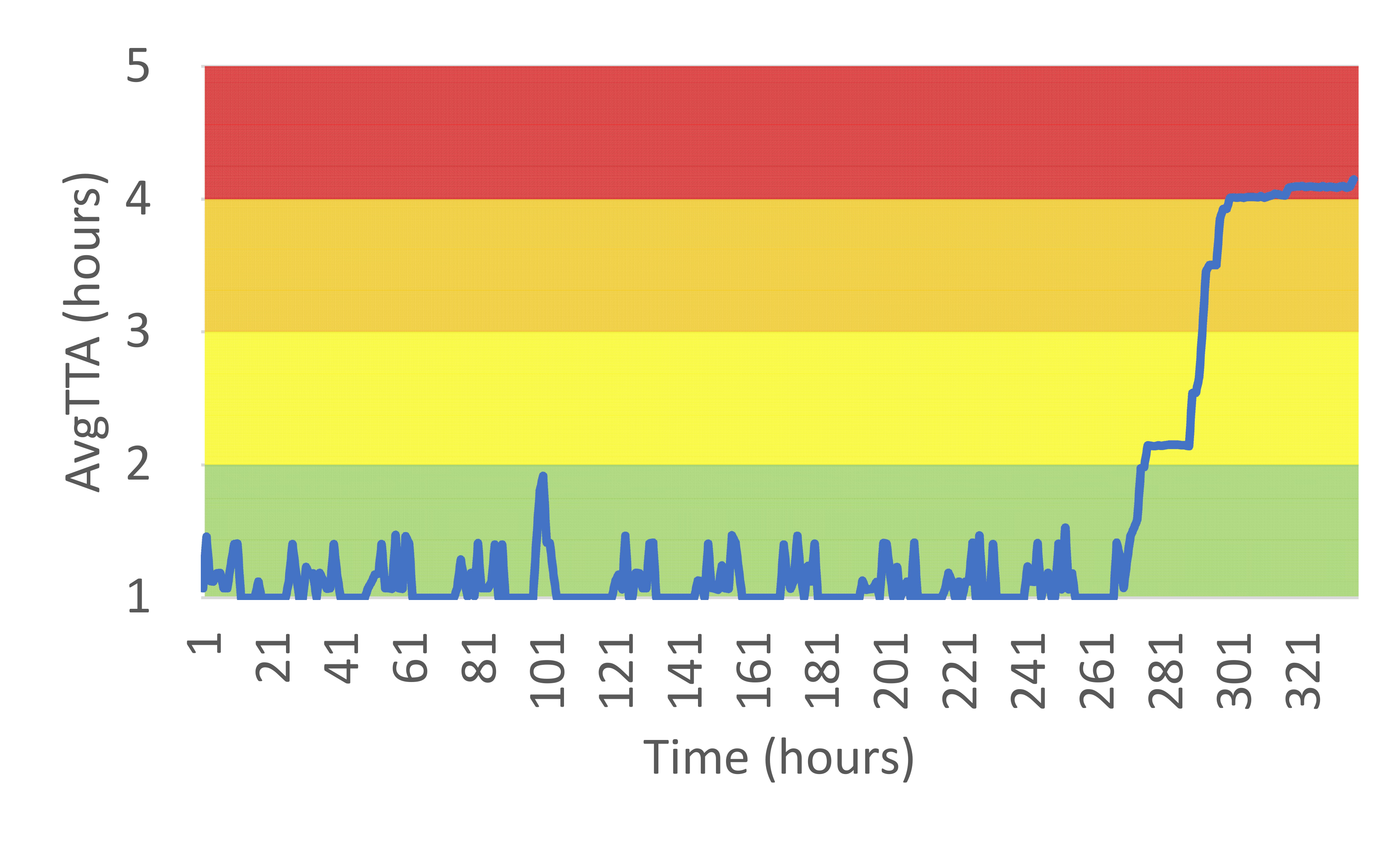}}%
    \qquad
    \subfigure[$S1$ policy proportions]{%
    \label{fig:tenth}%
    \quad\includegraphics[scale=0.17]{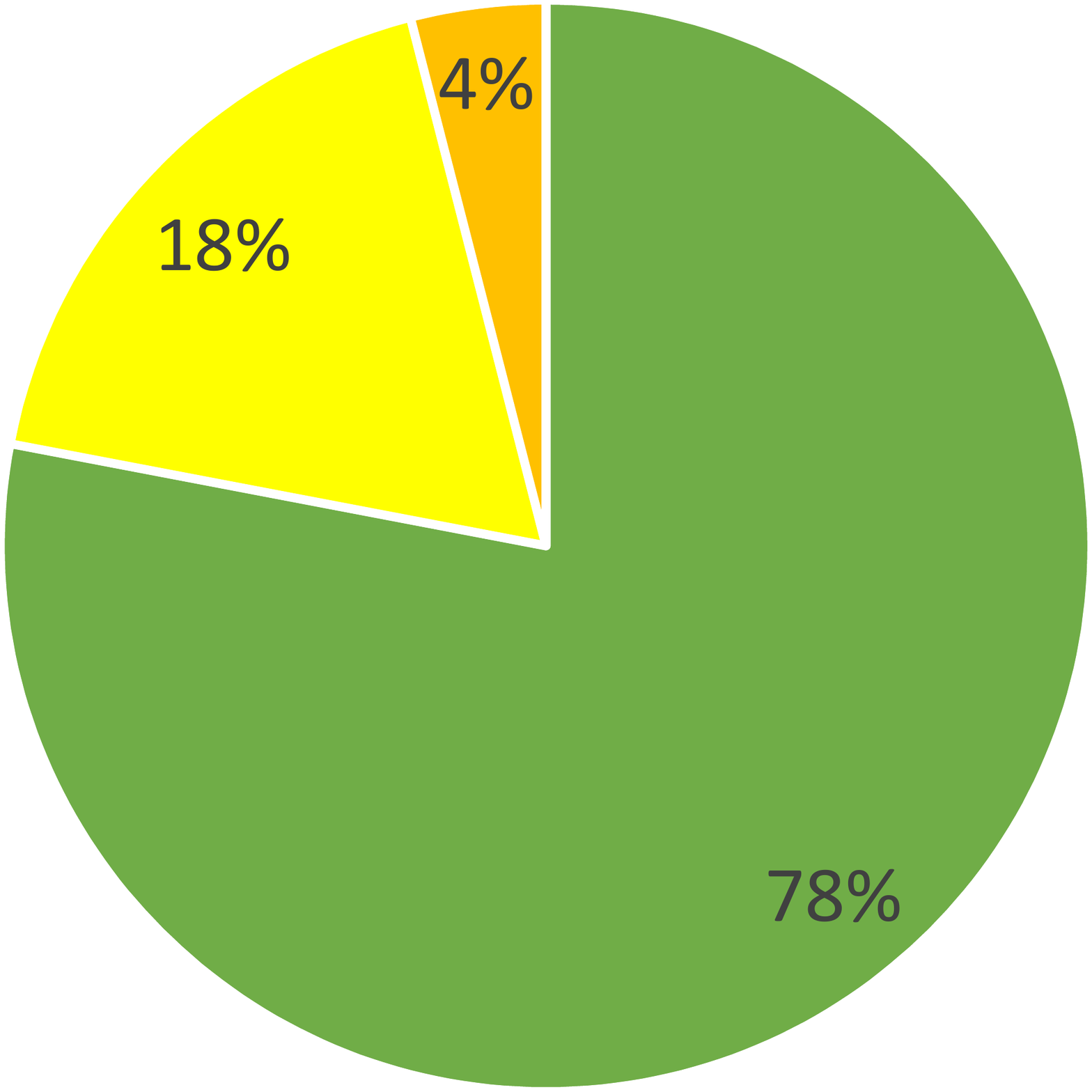}}%
    \qquad\qquad
    \subfigure[$S2$ policy  proportions]{%
    \label{fig:eleventh}%
    \includegraphics[scale=0.17]{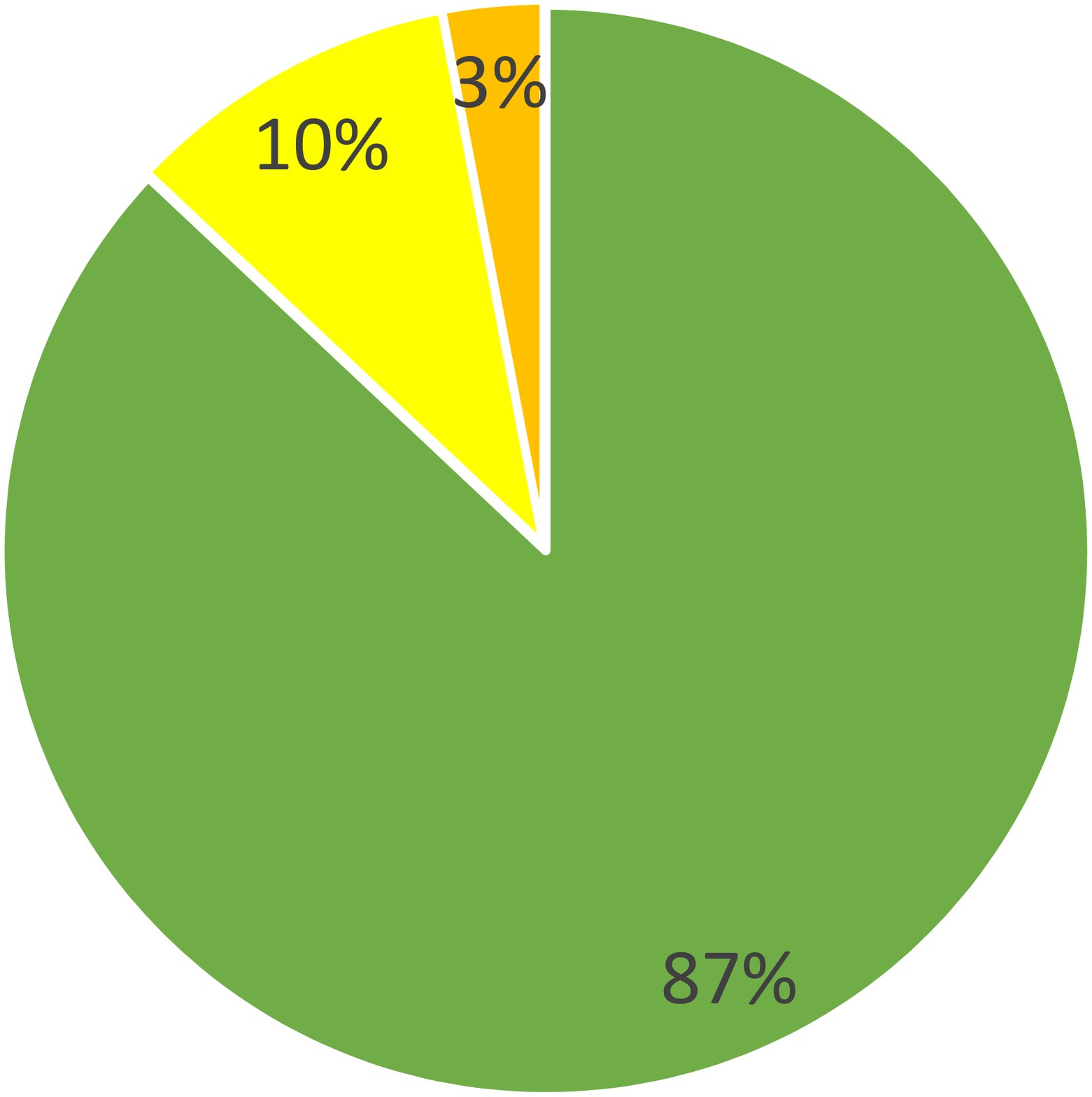}}%
    \qquad\;
    \subfigure[Model assumption attack proportions]{%
    \label{fig:twelth}%
    ~~\includegraphics[scale=0.17]{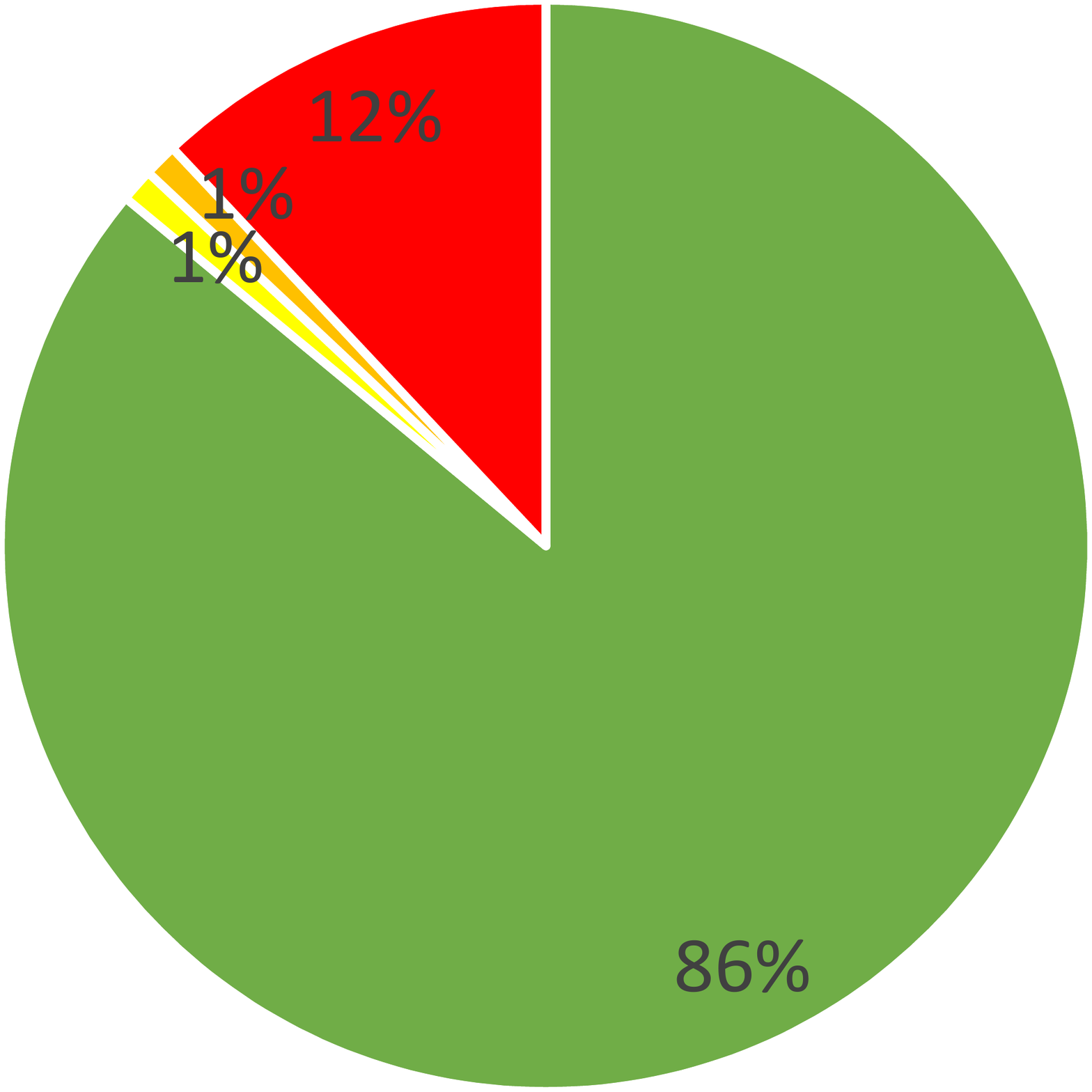}~~}%
\caption{Results for attack on baseline policies $S1$, $S2$ and attack against model assumption}
\end{figure*}

The CSOC-RL application uses a color-coded representation of AvgTTA/hr, which was developed in a prior work~\cite{Shah2018LOE}. Different color-coded tolerance bands are created  below the four hours and above the one hour
value of AvgTTA, for example as
shown in Fig.~\ref{fig:first}. As stated earlier, the AvgTTA value directly corresponds to amount of backlog (1 hour is 1175 alerts, 4 hours is 4350 alerts, and is linearly interpolated). 1-2 hours is green (acceptable), followed by yellow for 2-3 hour (acceptable), orange  for 3-4 hour (acceptable), and red (unacceptable) above 4 hours. Any value below 1 hour is rounded up to 1 hour. Our results show the hour by hour backlog for the worst run (run with maximum backlog ever) among 500 runs using the color-coded representation, for example, the line in Fig.~\ref{fig:first}. 

It is known that interpretability and usable interfaces are among the main issues in the adoption of AI technologies in the real world~\cite{ribeiro2016should,stein2017market}. Towards that end, the color-coded visualization was found to be extremely useful in displaying and explaining the CSOC-RL results to non-mathematical experts~\cite{Shah2018}. While  the visualization is coarser than numerical results expressed using the function $f$, they are much more easily interpretable by humans. 

In this work, we introduce another easily interpretable metric that we call the AvgTTA proportions. Given the stochastic arrivals of alerts due to the queuing process, it is not sufficient to look at a worst case run of the system. Thus, we show a pie-chart with the proportion of hours among the $500N$ hours over 500 distinct runs  corresponding to backlog in one of the four color bands.

\subsection{Adversarial RL}
We attack the learned defender RL using our adversarial RL method stated earlier. We try three different variations on the adversary budget: (1) adversary budget same as defender budget, but constrained by a daily bound $28800/14 \approx 2057 $ of additional alerts, (2) adversary budget same as defender budget with no bounds, and (3) adversary budget 10\% more than defender budget with  daily bounds.

The results for the worst case (maximum overall backlog) run of the system as well the AvgTTA distribution are shown in Fig.~\ref{fig:1}. The adversary with lower budget is unable to push the backlog beyond the green band in any of the 500 runs (\ref{fig:fourth}) with the daily bounds. Fig.~\ref{fig:first} shows that the attacker dumps all additional alerts allowed in a day at the start of the day. Even with no bounds, the attacker with lower budget is unable to push the backlog beyond the yellow band in the worst case (Fig.~\ref{fig:second}), with 97\% of $500N$ hours remaining in green band (Fig.~\ref{fig:fifth}). Again all additional alerts are dumped in a few hours at the start of the time horizon. However, with just an extra 10\% budget even the restricted attacker with a daily bound $1.1*28800/14 \approx 2262$ is able to push the backlog into the red zone in 22\% of the $500N$ hours (Fig.~\ref{fig:sixth}). Fig.~\ref{fig:third} shows that towards the end the backlog builds up and stays high since the defender resources get exhausted (see defender resources with time in appendix).
%The backlog does not reach the red zone always, because in those runs where the number of stochastic alerts is low the defender does not expend its additional resources, and is able to control the backlog.

These results provide an empirical evidence that the learned defender RL policy is robust to the best response of the attacker. Further, the results also reveal the subtle relation between robustness of defender and budgets of players as claimed in Theorem~\ref{thmonly}.

%State experimental setup. Three experiments - adversary with daily bound, adversary with no bound, adversary with 10\% higher budget (all with confidence bounds and $\mu$ varying).

\subsection{Theoretical Baseline}
Our next set of experiments are for the baselines $S1$ and $S2$ that we derived in Section~\ref{sec:game}. We choose the threshold $B$ corresponding to 2 hours. We experiment with the more powerful unbounded attacker (i.e., no daily bound). The 
worst case runs (Fig.~\ref{fig:seventh} and~\ref{fig:eigth}) show that the baselines perform reasonably and are able to recover from the initial barrage of alerts sent by the unbounded attacker.
Yet, the backlog still goes to the orange band, and compared to the proportions for RL policy (Fig.~\ref{fig:fifth}) the proportions are worse for the rule based policy (Fig.~\ref{fig:tenth} and~\ref{fig:eleventh}). However, the rule based policy is much simpler and interpretable than a large RL policy. In literature, there is evidence for the phenomenon that sub-optimal but simple decision aids are more convincing for human users and more likely to be trusted and adopted~\cite{elmalech2015suboptimal}. Thus, while sub-optimal, the rule based policy is still a competitive candidate for deployment.

%\textit{Remark}: The results show that the RL approach keeps the average backlog over $N$ hours low in addition to keeping the maximum backlog over $N$ hours low. This is because the RL value function with a discount factor close to 1 approximates the average backlog value (ignoring the $g$ term for RL reward). 
\subsection{Attacking Modeling Assumptions}
Finally, we show results for our attack that exploited modeling assumptions in the prior CSOC-RL work. Fig.~\ref{fig:ninth} shows that even the restricted attacker with a daily bound is able to successfully push the backlog into the red zone. We observe that the attacker is able to elicit a disproportionate response by defender, which results in the defender additional resources getting exhausted towards the end of the time horizon (see defender resources with time in appendix). Overall, the attacker is able to keep the backlog in the red zone in 12\% of the $500N$ hours (Fig.~\ref{fig:twelth}). The corresponding robust defender policy learned by retraining the defender by including episodes from the attack found has performance as shown in Fig.~\ref{defense}. Intuitively, the relearned defender policy exhibits more patience by allowing backlog to be higher on average (worst run in Fig.~\ref{defense}), yet manages to keep the backlog in green zone over the time horizon for 95\% of the $500N$ hours (piechart in Fig.~\ref{defense}). We also show the same attack against $S1$ and $S2$ in the appendix. We also attacked the relearned defender policy using various different chunk sizes for the adversary (all the way down to one) and found the defender to be robust against all of the chunk sizes less than 30.  

\begin{figure}[t]%
    \centering
    \subfigure{%
    \label{fig:thirteenth}%
    \includegraphics[scale=0.14]{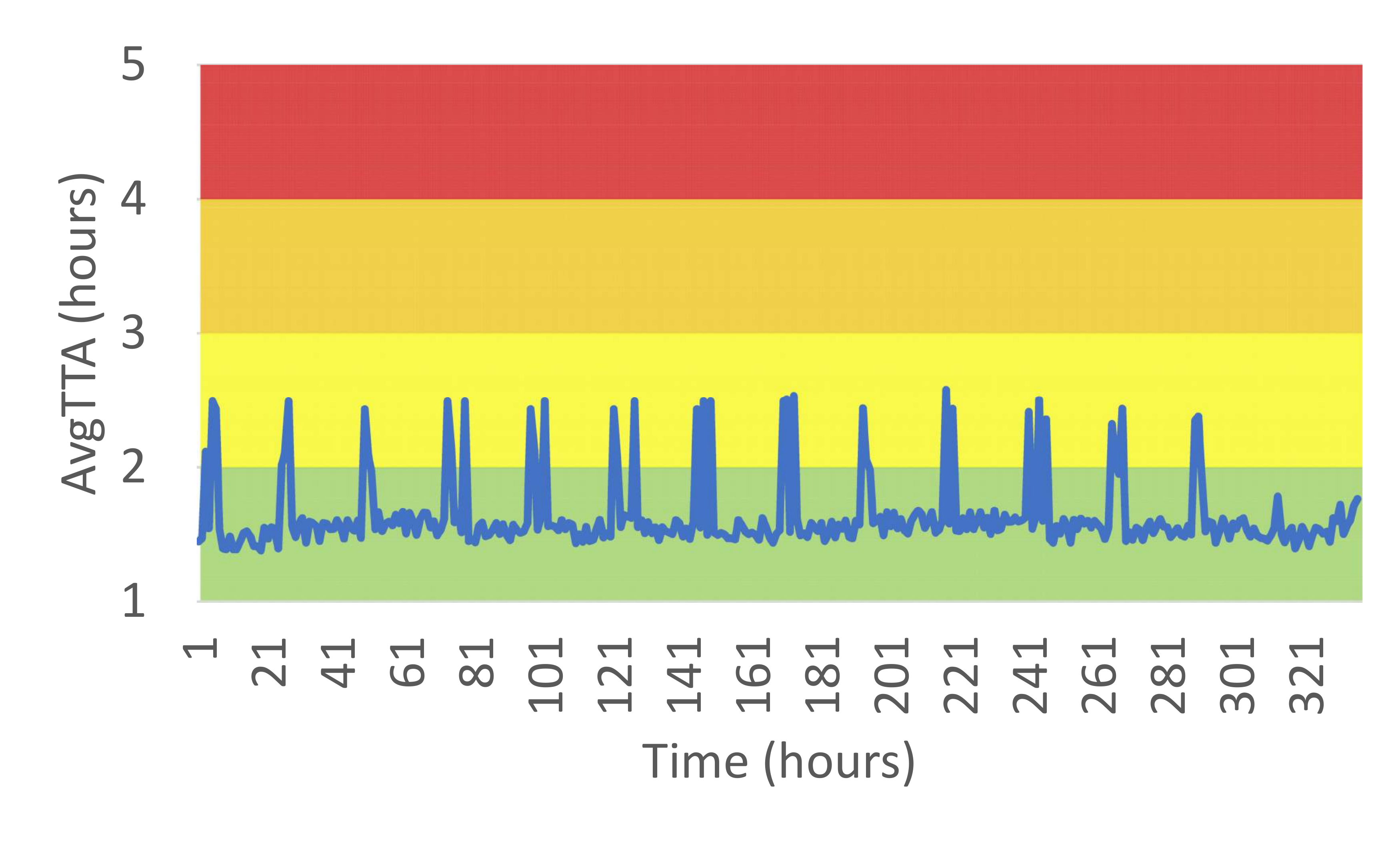}}%
    ~\subfigure{%
    \label{fig:fourteenth}%
    \includegraphics[scale=0.12]{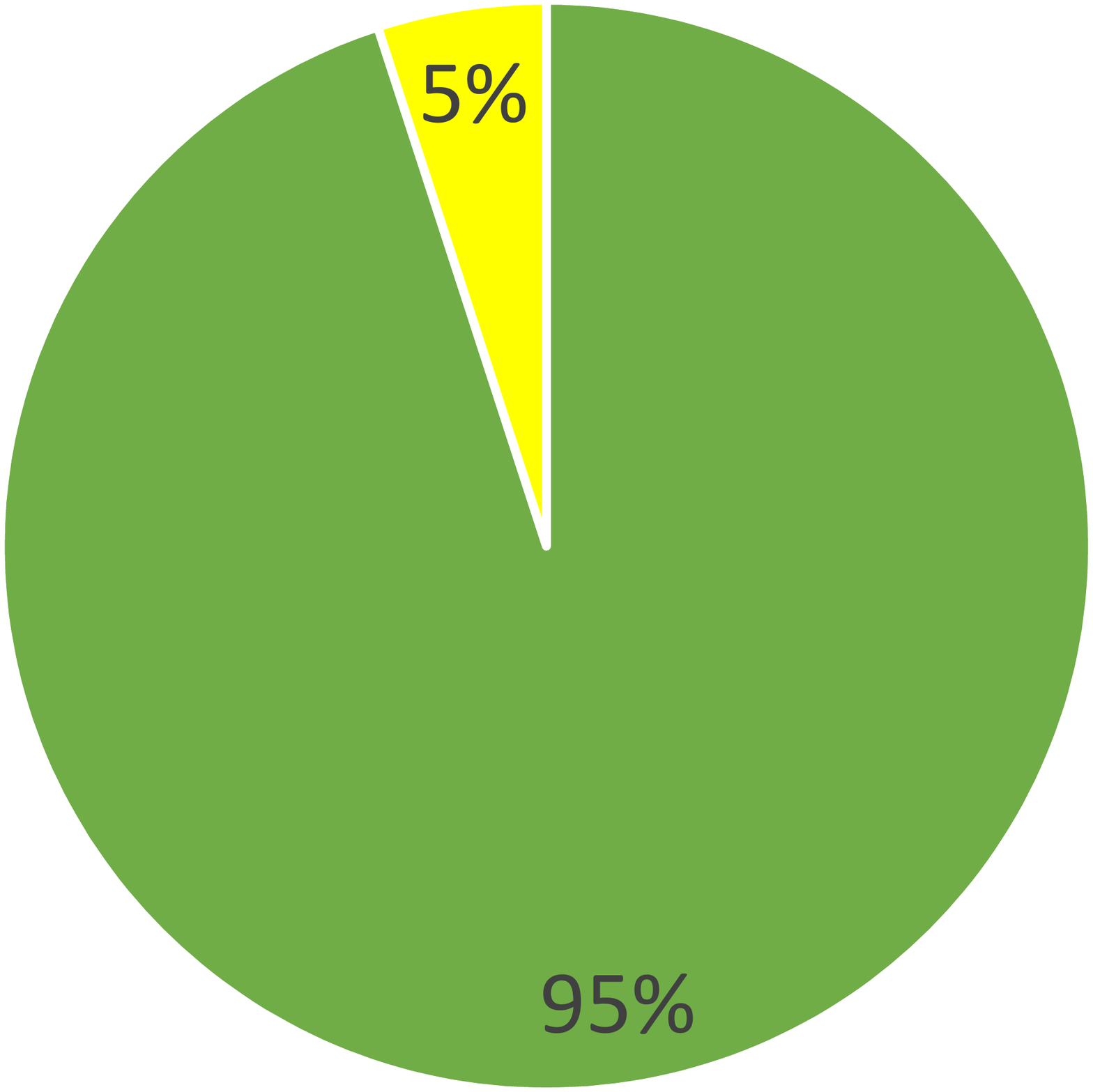}}%
\caption{Defense against attack on model assumptions. Left is worst case run and right is proportions.} \label{defense}
\end{figure}

\section{Related Work}
Game theoretic inspection or auditing has appeared in many papers~\cite{yan2018get,blocki2013audit,brown2016one}, however, all these works have a single shot interaction model and are not focused on cyber-alert inspection. As stated already, one game based work~\cite{schlenker2017don} on cyber-alert inspection is single shot, and has other stringent assumptions such as fixed number of analysts and alerts per hour. Another recent game theoretic work~\cite{guirguis2018} uses a zero-sum Markov game model for cyber-alert inspection, but the model assumes complete information for both the players and hence solves the game using standard minimax value iteration. Our approach accounts for the realistic scenario of incomplete information and signals, and explains theoretically and empirically the subtle relation between the budget of the players. Game theory has also been used for other problems in cyber-security such as deception~\cite{schlenker2018deceiving}, attack graphs~\cite{Durkota:2015:ONS:2832249.2832322}, man in the middle attacks~\cite{li2017defending}, and spear fishing~\cite{zhao2016optimizing,laszka2016multi}, which are quite different from our cyber-security problem. Also, scalability is still an issue in solving partially observable stochastic games~\cite{HBZaaai04}. Further, our game has a non-standard long term utility.

The previous CSOC-RL work~\cite{Shah2018} missed out on the strategic interaction by modeling the problem as a RL problem. Along same lines, recent work on adversarial attacks on learning techniques have found attacks on deep RL systems~\cite{huang2017adversarial,behzadan2017vulnerability,lin2017tactics}. A recent work~\cite{pinto2017robust} looks at the adversarial problem as a zero-sum stochastic game with complete information for both players and proposes a best response dynamics approach to solve the problem, that is, alternatively each player plays its best response fixing the other player's last policy. However, there is no guarantee of best response dynamics converging to an equilibrium, even in zero-sum games (for example, matching pennies). Our problem is harder due to the partial observation of current state. We utilize a different technique where the best response of a player is incorporated as episodes in the RL training of the other player, which is inspired by the double oracle technique in game theory. We also show that robust performance does not mean just reaching the equilibrium but also depends on the value of the equilibrium, which in our case relates to the resources (budget) of players.

Next, queuing processes are well established as a natural model for arrivals and service~\cite{cohen2012regenerative,cohen1982single}. Some work extend the classic queuing processes to deal with multiple rational customers who can opt-out of joining a queue~\cite{Altman2005} or strategic selection of scheduling criteria such as FCFS~\cite{ashlagi2013equilibria}.
However, our model has a adversarial interaction of two players on top of the queuing process, which, as far as we know, has not been addressed in the queuing theory literature.

\section{Conclusion}
We performed a red team evaluation of a cyber-alert inspection system that is under consideration for deployment. We observed the system to be robust to the best responding adversarial alert generation policy, but we also showed that a weakness in the model assumption could be exploited by an attacker. We provided a game theoretic formulation of the problem, which allowed us to understand the defender-adversary interaction. In particular, we showed that players' resources (budget) are a critical factor in deciding whether the defender policy can be robust. The theory also yielded simple, and  sub-optimal but usable defender policies. Further, using game theoretic insight we made the defender RL approach robust to the discovered successful attacker policy. The adversarial RL and double oracle approach in RL are general techniques that are applicable to other RL usage in adversarial environments.

\section{Acknowledgement}
This work was supported in part by the Army Research Office under MURI grant W911NF-13-1-0421.

\clearpage

\bibliography{aaai19}
\bibliographystyle{aaai}

%\end{document}
\clearpage

\appendix
\section*{Proof of Theorem~\ref{thmonly}}
\begin{proof}
First, the case when $Y = X + 4800$. Consider the attacker sending $B = 2400$ alerts every hour from the start. In 14 hours, the attackers sends $28,800 + 4,800 = 33,600$  additional alerts. In 14 hours, the number of arrivals $S$ from the Poisson process form a Poisson distribution with rate $14 \lambda$. The following concentration inequality is known~\cite{pollard}
$$
P(S \leq 14 \lambda - 0.3 \lambda) \leq \exp(- \frac{0.3^2 \lambda^2}{2 \times 14 \lambda})
 \leq 0.002$$
Thus, with probability $\geq 0.998$ the number of arrivals in 14 hours more than $13.7\lambda \geq 26,290$. The number that can be served by normal analysts in 14 hours is $14 \times 1920 = 26,880$, that is 590 of the additional alerts could be served by normal inspections. Thus, by the 14$^{th}$ hour even if the defender has used all additional resources $X$, there are still $33,600 - 28,800 - 590 = 4210$ alerts in the queue, which corresponds to $f(4210) = 0.955$. Thus, the utility is worse than $- 0.955*0.998 = - 0.953$  

Next, we consider the case when $Y < X$. In this case we first prove that the probability of the backlog being more than $B$ over the $N$ finite horizon is small.

A busy cycle for a M/D/1 queue is defined as a time period in which the queue length first goes to 0 starting from 0. It is known that all busy cycles are i.i.d., since a M/D/1 queue is a regenerative process~\cite{cohen1982single}. A busy cycle with time length less than $B/\mu$ hour will not have more than $B$ alerts in the queue, since by definition the server is busy and it can serve $< B$ alert in less than $B/\mu$ hour, so the max queue length cannot be $\geq B$. Thus, we focus on busy cycles of time length $\geq B/\mu$ hour. There can only be $N\mu/B$ such busy cycles in $N$ hours. 

In each busy cycle the probability of the queue length being at most $j$ is given by $P_j$, where it is known that $P_j \geq 1 - \frac{1}{j}$~\cite{cohen1982single}. Due to i.i.d. busy cycles the probability of queue length being less than $B$ over all the $N$ hours is then more than $(P_B)^{N\mu/B}$. 
%For our case, with $A=1800$, this value is $\geq 0.81$.

Next, consider the event where the queue length remains below $B$ for all $N$ hours. Consider one run of the underlying queuing process with the queue lengths $x_1$ to $x_{336}$ at every hour. Any additional alerts $k$ sent by the adversary in the $j^{th}$ hour raises the subsequent queue lengths to $x_{i} + k$ for $i \geq j$. Let the adversary send alerts quantity $k_1, k_2, \ldots$ at $j_1, j_2, \ldots$ respectively. According to the defender's policy the queue length will always remain below $B$ after the defender action as long as the defender has enough resources to allocate. Thus, we will argue that the defender resources are not exhausted. Consider the maximum queue length without defender intervention between $j_1$ and $j_2$: $Q = \max\{x_i + k_1~|~j_1 \leq i < j_2\}$. Let the total defender intervention within $j_1, j_2$ be $d_1$. At any time, if the total defender intervention within $j_1, j_2$ is equal to $Q - B$, then the queue length remains below $\max\{x_i~|~j_1 \leq i < j_2\} \leq B$ beyond this point. The defender intervention also is never more than $Q - B$ because if so then the last time that the defender intervenes, he can reduce his resource allocation to be exactly $Q - B$ and achieve the goal of keeping queue below $B$. Thus, $d_1 \leq Q - B$. Next, since each $x_i \leq B$, we have $k_1 \geq Q-B \geq d_1$. In a similar manner, it can be seen that $k_i \geq d_i$ for all $i$. Thus, the total additional resources used by defender will always be less than the additional alerts sent by adversary. Since the budgets are same, the defender will not run of resources.

Given the maximum queue length remains below $B$, providing defender utility $-f(B)$ with probability more than $(P_B)^{N\mu/B}$ and assuming the worst case utility of $-1$ otherwise, we obtain the expected reward: $(1 + f(B) )(P_B)^{N\mu/B} - 1$
\end{proof}

\section*{Analyst Scheduling and Experimental Setup}
We directly quote from the previous CSOC-RL work~\cite{Shah2018}:

\noindent\textbf{Extra inspections}: On page 8 
\begin{quote}
``It is assumed that all analysts spend 80\% of their effort in a shift toward
alert analysis, and the rest of the time is spent on report writing, training, and on
generating signatures. Hence, an analyst could increase their effort on alert analysis
up to 20\% when the need arises, which will increase the service rate of alerts investigated
in a day.''
\end{quote}
Then again on page 9: 
\begin{quote}``When a disruptive event occurs, a CSOC manager, in the order of preference as
determined through our discussions with CSOC managers at the Army Research Lab,
would utilize the remainder 20\% of analyst time on alert analysis, spend some of their
own time to assist the analysts in clearing the alert backlog, and bring on-call analysts
to supplement the regular analyst workforce.''\end{quote}

\noindent\textbf{Experimental Setup}: First, we present the setup for the M/D/1 queue. This quote below is from page 13; the equations and citations in the quoted text refer to ones in the CSOC-RL paper; more details are in the CSOC-RL paper. 
\begin{quote}
``There are 10 clusters of sensors considered for the experiment. Typically, each cluster
has about 10 to 12 sensors. Based on Equation (5) (see Electronic Appendix for
equations), the average arrival rate $\lambda$ = 1,919 alerts per hour, and based on Equation
(6), the average service rate $\mu$ = 1,920 alerts per hour for the entire system (U = 0.8),
which combines all the sensors and analysts in the system. Hence, traffic intensity
$\rho < 1$ for the baseline case scenario, where $\rho = \lambda / \mu$ . It should be noted that the service
rate is based on 80\% effort spent on alert analysis, and at this effort value, the
alert service rate is slightly greater than the alert arrival rate (prevents the formation
of an infinite queue length of alerts). A larger difference between $\lambda$ and $\mu$ ($\lambda <\!< \mu$) at U = 0.8 would suggest that the analysts have idle time. Hence, it is customary
at the CSOC to set the value of $\lambda$ close to $\mu$ at U = 0.8 ($\rho < 1$), and deal with $\rho \geq 1$
situations by utilizing a) the remaining 20\% of analysts time on alert analysis, b) portions
of manager’s time, and c) on-call analysts if needed. The above parameter values
were obtained through discussions with the CSOC managers of both distributed and
centralized CSOCs [Zimmerman 2014], however, different values of $\lambda$ and $\mu$ could be
used to set the baseline value of avgTTA/hr for a CSOC with $\rho < 1$. After a few days
of CSOC operations, upon reaching a steady state under normal operating conditions,
the acceptable value of avgTTA/hr was determined to be 1 hour and the nominal average
queue length was 1,175 alerts at any given point in time.''\end{quote}
\begin{figure}[t]%
    \centering
    \subfigure[Defender resources for Fig.~\ref{fig:third}]{%
    \label{fig:extraone}%
    \includegraphics[scale=0.2]{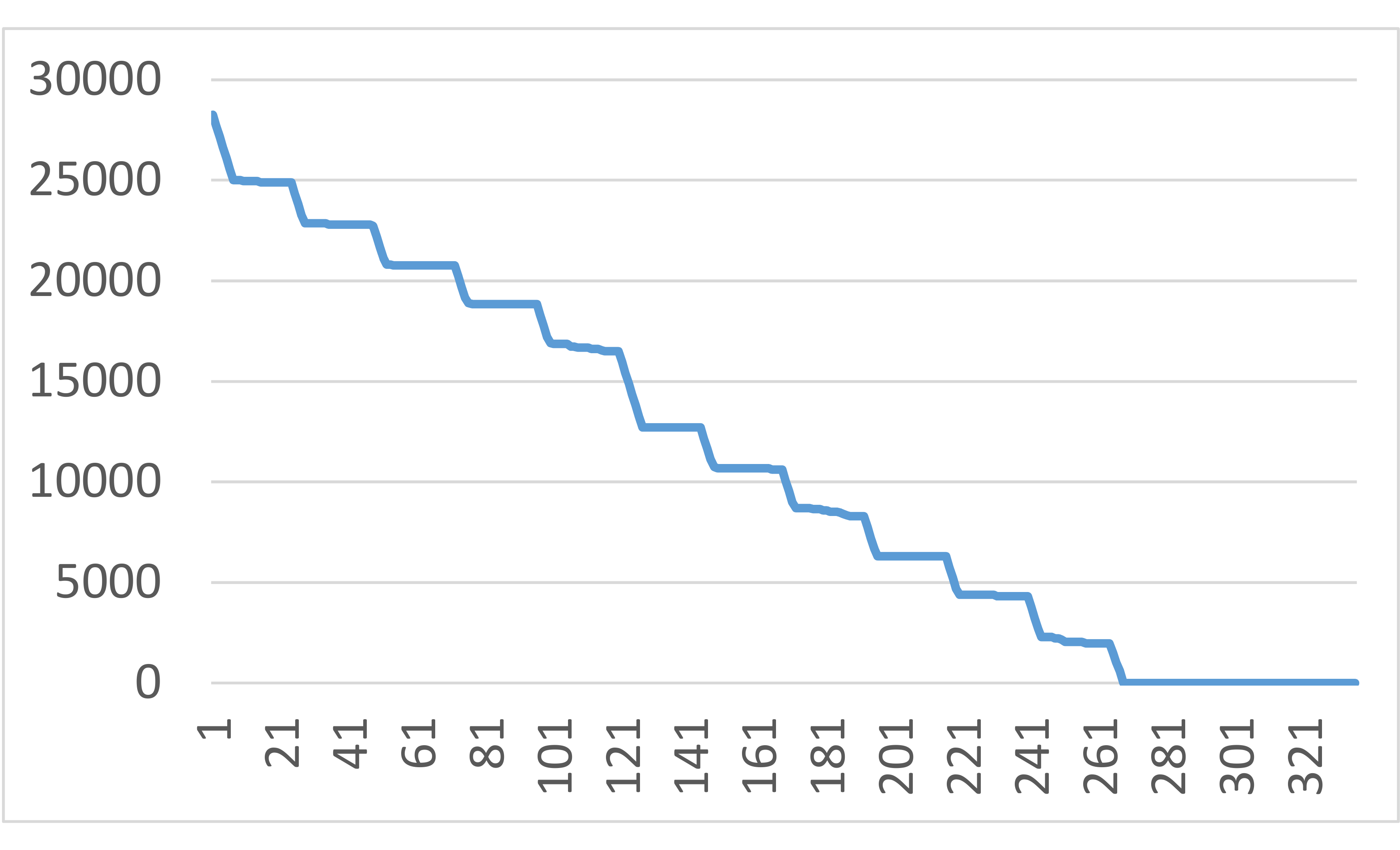}}%
    \quad
    \subfigure[Defender resources for Fig.~\ref{fig:ninth}]{%
    \label{fig:extratwo}%
    \includegraphics[scale=0.2]{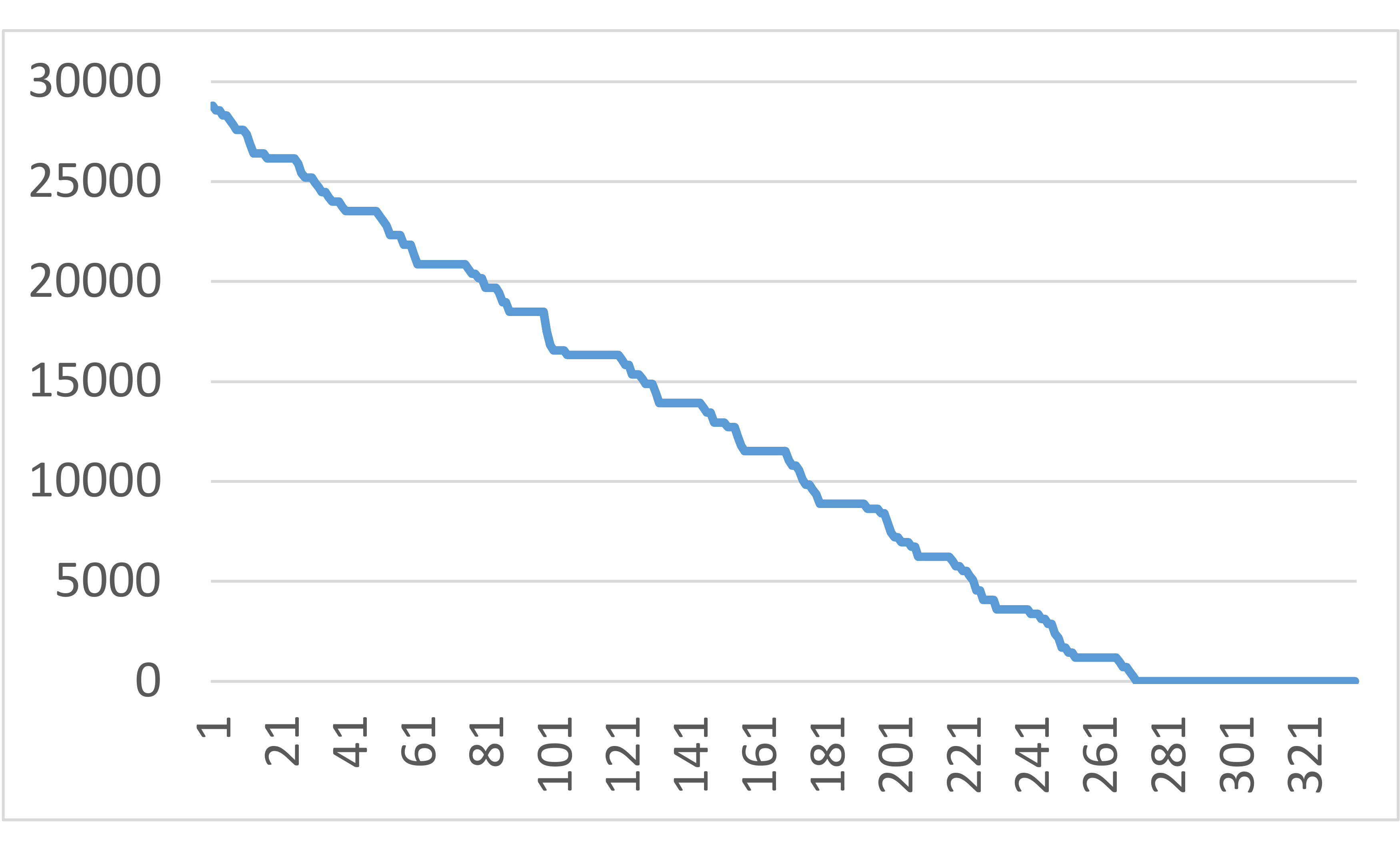}}%
\caption{Remaining defender resources over time}\label{fig:extra}
\end{figure}
\section*{Defender Remaining Resources Results}

Figure~\ref{fig:extra} shows the defender resources corresponding to the runs in Figure~\ref{fig:third} and~\ref{fig:ninth}. In both cases the defender resources get exhausted before the $N$ hours, and then the attacker is able to successfully push the backlog high.

\section*{Modeling assumption attack for $S1$ and $S2$}
Figure~\ref{fig:extraattack} shows the results of attacking $S1$ and $S2$ using the restricted daily bounded adversary with a finer discretization. It can be seen $S1$ is quite robust to the attack, whereas the more aggressive $S2$ suffers from the attack, which highlights our observation that a more patient policy is beneficial against this attack. For $S2$ the defender resources get exhausted before the $N$ hours, and then the attacker is able to successfully push the backlog high. Note that the piechart for $S1$ looks better than Fig.~\ref{fig:tenth} because this attack used a daily bounded adversary, unlike the unbounded adversary attack on $S1$ and $S2$ in Fig.~\ref{fig:tenth} and~\ref{fig:eleventh} respectively. 
\begin{figure}[t]%
    \centering
    \subfigure[Modeling assumption attack ($S1$) worst run]{%
    %\label{fig:extraone}%
    ~~\includegraphics[scale=0.18]{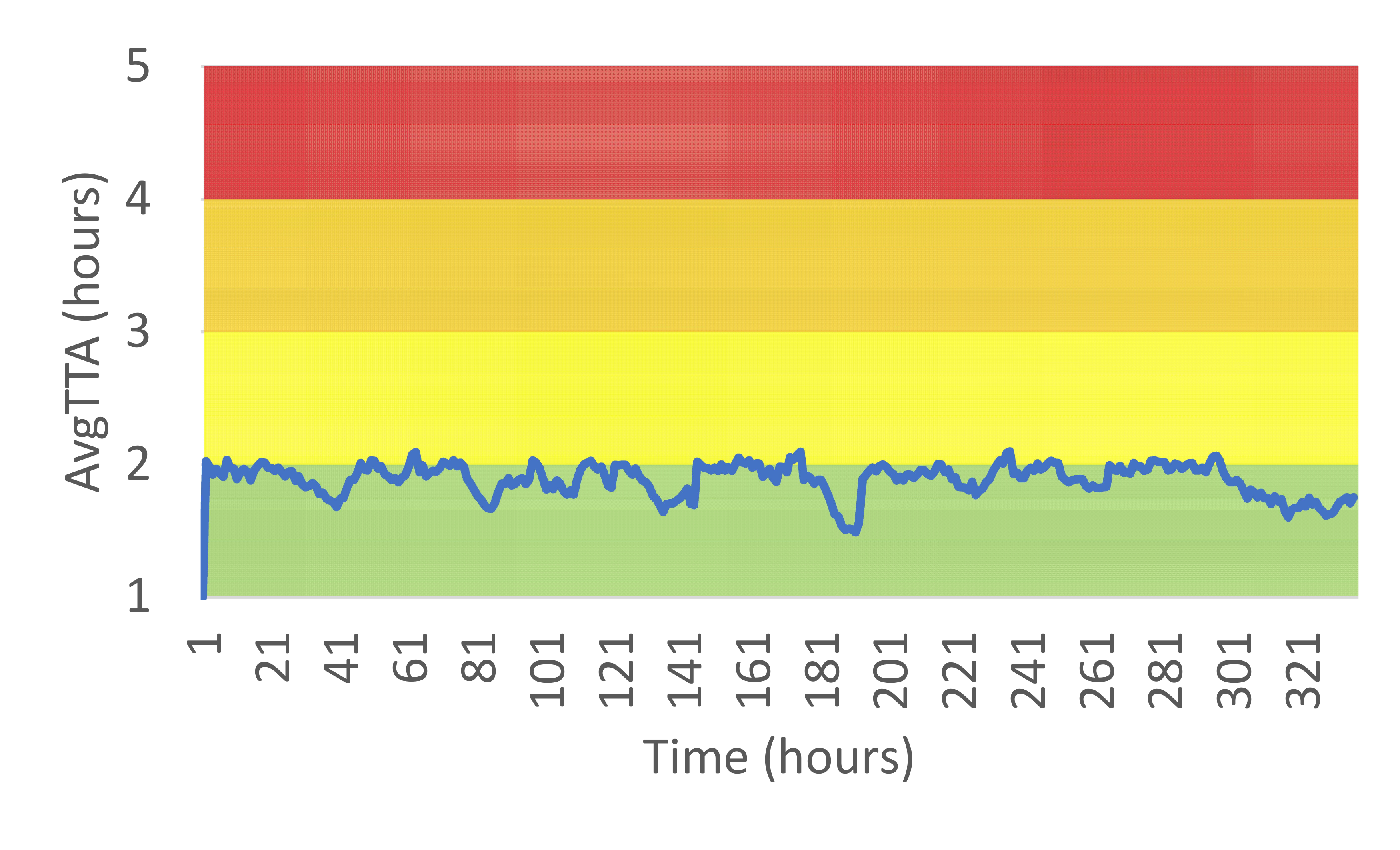}~~}%
    \quad
    \subfigure[Modeling assumption attack ($S2$) worst run]{%
    \label{fig:extratwopie1}%
    ~~\includegraphics[scale=0.18]{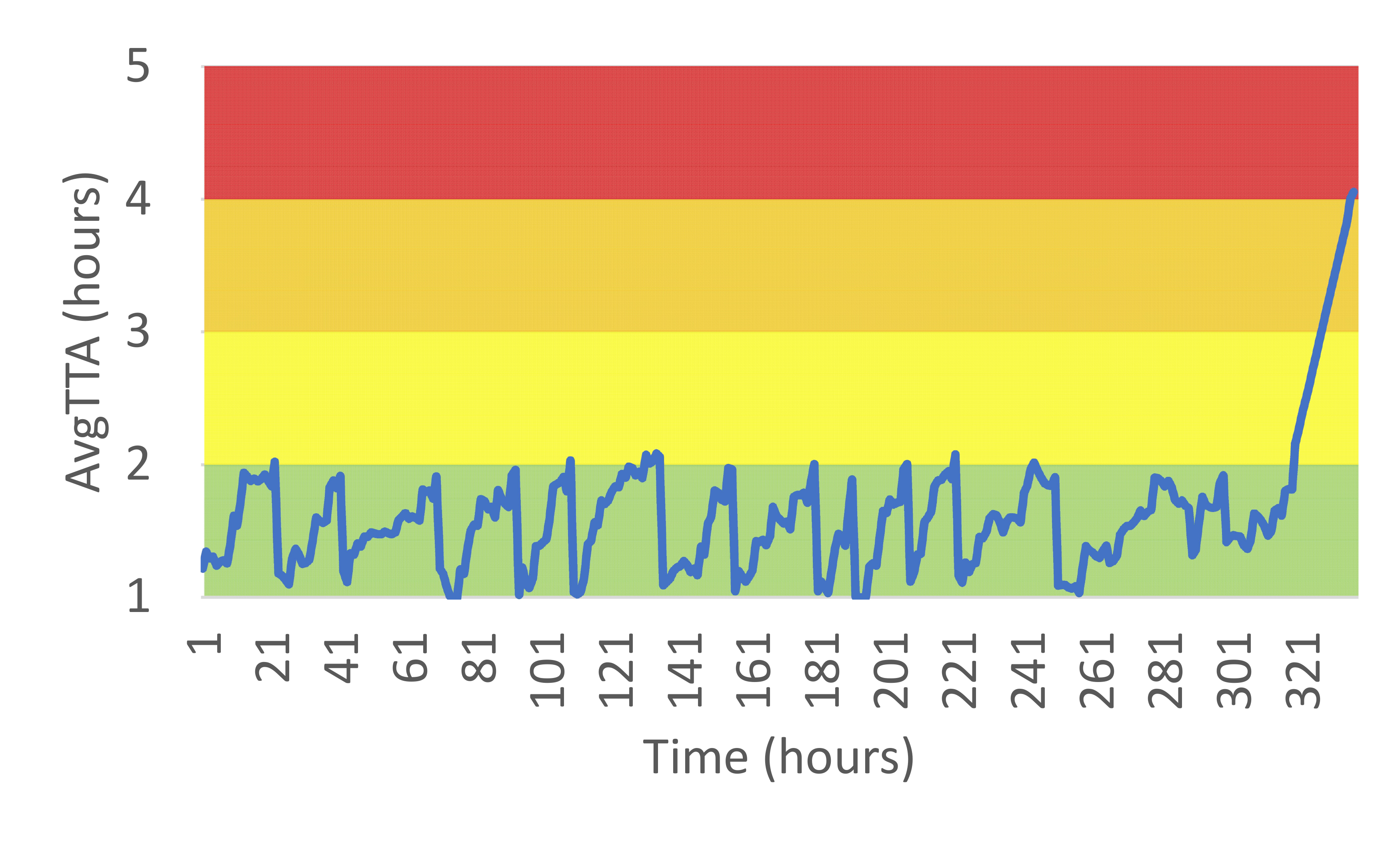}~~}%
    \quad
    \subfigure[Modeling assumption attack ($S1$) proportions]{%
    %\label{fig:extraone}%
    \includegraphics[scale=0.22]{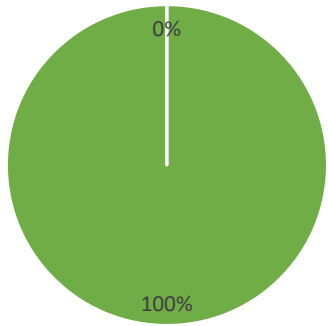}}%
    \quad%
    \subfigure[Modeling assumption attack ($S2$) proportions]{%
    \label{fig:extratwopie2}%
    \includegraphics[scale=0.22]{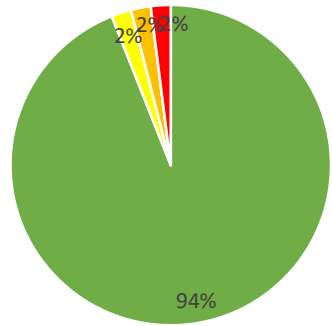}}%
\caption{Modeling assumption attack on $S1$ and $S2$}\label{fig:extraattack}
\end{figure}

\end{document}